\documentclass[12pt,a4paper]{article}
\usepackage[onehalfspacing]{setspace}
\usepackage[margin=1.25in]{geometry}
\usepackage{amsmath,mathtools,amssymb,amsthm}
\usepackage[authoryear]{natbib}
\usepackage{authblk}
\newcommand{\bm}{\boldsymbol}
\usepackage{algorithm}
\usepackage{algpseudocode}
\usepackage{caption,subcaption,tabularx,booktabs}
\usepackage[]{threeparttable}
\usepackage{multirow}
\usepackage{rotating}
\usepackage{xcolor}
\usepackage{url}
\usepackage[normalem]{ulem}

\DeclareMathOperator*{\argmax}{arg\,max}

\newcommand{\dt}{\mathrm{d}}
\newcommand{\E}{\mathbb{E}}
\newcommand{\wh}{\widehat}

\newtheorem{theorem}{Theorem}

\newtheorem{assumption}{Assumption}

\newtheorem{lemma}{Lemma}

\usepackage{comment,color}
 \includecomment{comment}
 \specialcomment{comment}{\begingroup\color{red}}{\endgroup}  %%%

\title{Scalable likelihood-based inference for limited dependent variable models
}
\date{\today}
\author{David T. Frazier}
\author{Rub\'en Loaiza-Maya}
\author{Didier Nibbering\thanks{Correspondence to: Department of Econometrics \& Business Statistics, Monash University, Clayton VIC 3800, Australia, e-mail: \textsf{didier.nibbering@monash.edu}}}
\affil{\small Department of Econometrics and Business Statistics, Monash University}

\begin{document}
\maketitle
\begin{abstract} % ecta and restud max 150 words
\noindent \footnotesize
% SEGA: Stochastically estimated gradient ascent.
% Asymptotically valid.
% SGA using Fisher to estimate gradient. 
Limited dependent variable models are central to empirical economics, but likelihood-based inference is infeasible when likelihoods involve high-dimensional integration over latent variables. This paper proposes Stochastically Estimated Gradient Ascent (SEGA), a scalable estimation approach for limited dependent variable models. Using Fisher's identity, SEGA replaces the intractable likelihood score with an unbiased augmented data score evaluated at a single conditional draw of the latent variables, and embeds this score in a stochastic gradient ascent algorithm. With sufficiently many iterations, we show that SEGA is asymptotically equivalent to the infeasible maximum likelihood estimator. A variance estimator based on Fisher's and Louis' identities is proposed that allows inference to proceed in the usual manner. Applications to brand choice and household demand demonstrate the usefulness of SEGA for conducting inference in large-scale discrete choice and censored demand models.
\end{abstract}

\thispagestyle{empty}
\clearpage
\setcounter{page}{1}
\section{Introduction}
% KEY THINGS TO CONSIDER:
% Why should an applied economist care? 
% What new empirical work becomes possible? 
% What is the economic payoff beyond a faster estimator?
% we remove a practical barrier to likelihood-based inference in economically important latent-variable models.

% BUT: for econometrica we also need to create some appreciation for the theory, and clearly point out its novelty

%\dn{LDVMs important in general interest journals}\\
Limited dependent variable models (LDVMs) are central to empirical economics because they allow researchers to model outcomes that are discrete, censored, truncated, or otherwise restricted by their support \citep{tobin1958estimation}. Such restrictions arise naturally in settings where researchers study a range of decisions made by economic agents, such as choices, participation, purchase incidence, sales, or program uptake. Recent applications of LDVMs appear in labor economics \citep{bonhomme2022discretizing}, development economics \citep{bhattacharya2024demand}, health economics \citep{pakes2021unobserved}, education economics \citep{kline2016evaluating}, marketing \citep{danaher2020advertising}, and digital economics \citep{kanazawa2026ai}. 

%\dn{But still big issue with applicability of LDVMs}\\
%Yet inference in many interesting applications 
While LDVMs can be used to model data arising from many economic settings, inference within this class of models remains challenging. On the one hand, the limited observability of the dependent variable means that observations only partially reveal the economic behavior of interest. As a consequence, large datasets are required to infer unknown parameters of interest with enough precision to be economically informative. On the other hand, estimation of such models is cumbersome even in small datasets since the resulting likelihood functions for many LDVMs involve high-dimensional integrals that are unavailable in closed-form. Consequently, for the large datasets necessary to produce economically useful inferences, applying maximum likelihood estimation (MLE) to conduct inference using LDVMs is often practically infeasible. To circumvent these issues, researchers are often forced to impose restrictive structures and to use approximate methods without the inferential guarantees associated with MLE.

% Paragraph 3: Main idea.
%\dn{What do we do, and how does this solve the problem?}\\
This paper removes the need to calculate the high-dimensional integral required to evaluate the likelihood function, enabling inference in economically important LDVMs. To this end, we propose Stochastically Estimated Gradient Ascent (SEGA), a scalable estimation method for LDVMs that admit a latent-variable representation. The key idea is to use Fisher's identity to express the score of the likelihood as a conditional expectation of the latent variable augmented score \citep{Cappe}. It follows that \textit{a single draw} from the conditional distribution of the latent variables provides an unbiased estimate of the score, irrespective of the sample size or dimension of the parameter space. We embed this noisy score in a stochastic gradient ascent (SGA) algorithm, and smooth the iterates from this algorithm to obtain a low variance estimate. For inference, we combine Fisher's identity with Louis' identity to construct a sandwich variance estimator that avoids direct evaluation of the observed score and Hessian \citep{louis1982finding}. SEGA is widely applicable because many LDVMs used in economics have the three features needed for fast implementation: a latent-variable representation, an augmented likelihood with an analytically tractable score, and  latent variables that can be generated efficiently from the conditional distribution.

% Paragraph 4: Empirical payoff.
% Large-scale MNP for brand substitution and random-effects Tobit for household demand,
% both with valid frequentist inference at scales where standard likelihood methods are
% not feasible.
%information accrues slowly with discrete choices or censoring, so large data essential to estimate substition patterns or price effects.
%\dn{What new empirical work becomes possible. Connect to paragraph two.}\\
We illustrate the empirical payoff of SEGA with two large-scale applications. The first is a multinomial probit (MNP) model of pasta brand choice using more than one million purchase observations. In this setting, a brand purchase identifies the chosen alternative but not the full ranking of alternatives, so large samples are needed to recover substitution patterns. SEGA allows us to estimate the unrestricted model, construct confidence intervals for own- and cross-price elasticities, and formally test covariance restrictions that are commonly imposed for tractability. The second application is a random effects Tobit model of household-level pasta demand using more than 200,000 households observed over 62 weeks. In this setting, observed sales record whether demand is positive but not latent desired demand when purchases are censored at zero. SEGA allows us to study heterogeneity in baseline demand, price sensitivity, and state dependence, and to test whether these sources of heterogeneity are independent. In both applications, standard likelihood-based approaches are infeasible at the required scale. SEGA therefore expands the class of empirically relevant LDVMs for which researchers can conduct frequentist likelihood-based inference.

%Paragraph 5: Theory and numerical validation.
%Finite-sample concentration around the infeasible MLE; asymptotic equivalence to the
%MLE when iterations grow; sandwich inference under misspecification.
%\dn{Make sure that econometrica appreciates the theoretical contribution}\\
We make two main theoretical contributions. First, for fixed sample sizes, we prove that the SEGA iterates concentrate around the infeasible MLE as the number of iterations increases. Second, when the number of iterations grows sufficiently fast with the sample size, we show that SEGA is asymptotically equivalent at first order to the infeasible MLE, despite using only a single simulated draw of the latent variables at each iteration and under arbitrary starting values on a compact parameter space. The latter is particularly important as an initially consistent estimator is difficult to obtain in practice for large LDVMs. The novelty in this contribution comes from the fact that, in contrast to standard SGA algorithms, simulation noise in SEGA is generated by latent variable draws from a conditional distribution that changes with the current parameter value, and not from subsampling observations \citep{polyak1992acceleration,moulines2011non,chen2020statistical}. While the noisy score we use to implement SEGA is path-dependent, and must be handled appropriately, Fisher's identity ensures that the estimated score is an unbiased estimate of the infeasible score that we leverage to obtain our theoretical results. To our knowledge, this is the first result showing that simulated latent variable scores can be used inside a stochastic optimization algorithm while remaining asymptotically equivalent to the infeasible MLE.

SEGA is closest in spirit to the method of simulated scores (MSS) of \citet{hajivassiliou1998method}. Despite its attractive theoretical properties, MSS has long been known to be difficult to implement in practice \citep{cameron2005microeconometrics,train2009discrete}. MSS approximates the likelihood score using simulated latent variables and then directly solves the resulting simulated score equations. This requires the simulated score function to remain stable across candidate parameter values, which is challenging in models where the latent variables are constrained by the underlying economic model, such as in high-dimensional choice models. % where the latent variables are latent utilities associated with choices. 
Our numerical experiments directly demonstrate the sensitivity of MSS to the number of simulation draws, the accuracy of the conditional simulator, and the initialization of the algorithm, generating substantial simulation noise and numerical instability even in moderate-dimensional LDVMs.
SEGA provides an alternative approach: rather than eliminating simulation noise before solving the score equations, it leverages stochastic approximation and iterative averaging to curtail the simulation noise across the optimization path.

Similar to SEGA, the contemporaneous approach of \cite{lee2026sauss} also uses a stochastic approximation algorithm to avoid the implementation issues of MSS, but ultimately delivers an estimator with a different asymptotic distribution to that of the infeasible MLE. Proposed in the confines of MNP models, the approach of \cite{lee2026sauss} uses a different unbiased gradient estimator that is combined with accept-reject simulation steps,  mini-batching, and multiple simulated draws of the latent variables per algorithm iteration, to provide parameter estimates. The resulting algorithm  requires a fraction of the computational cost associated with MSS. Given an initially consistent starting estimator, the authors develop an inference framework that explicitly accounts for variation due to mini-batching and a finite number of latent-variable draws.   

% therefore, practtioners have used alternative methods in practice, such as sml or quadrature, taking theoretical drawbacks for granted. 
%\dn{Discuss methods used in practice: relate to Sections 2.2.1, 2.2.2, and 5.3}\\
In practice, applied researchers often use simulated maximum likelihood (SML) or numerical quadrature to approximate the integrals in LDVM likelihoods. SML replaces likelihood contributions, such as the MNP choice probabilities, with Monte Carlo approximations \citep{borsch1993smooth}. The simulated score is biased for a fixed number of simulation draws. Consistency therefore requires the number of simulation draws to increase with the sample size. This requirement becomes more demanding in MNP models as the number of alternatives grows, because the choice probabilities involve higher-dimensional integrals and more simulation effort is needed to maintain accuracy. Numerical quadrature is commonly used in random effects LDVMs \citep{rabe2005maximum}. Its cost often grows exponentially with the dimension of the random effects, making it impractical with several sources of unobserved heterogeneity. As a result, SML and quadrature methods are feasible only in low-dimensional LDVMs, or in settings where strong model-specific restrictions make the integrals manageable. SEGA removes this limitation by avoiding accurate likelihood approximation at each parameter value.

% Less often used with LDVMs, but in principle possible approach, is EM.
% EM is possible, but exact EM is generally unavailable, and approximate EM becomes costly
%\dn{Discuss EM methods}\\
Although less commonly used with LDVMs, expectation-maximization (EM) is another possible method for obtaining point estimators \citep{dempster1977maximum}. Rather than approximating the likelihood directly, EM algorithms alternate between computing conditional expectations of the latent variables given the observed data and current parameter values, and maximizing the resulting expected augmented likelihood. In LDVMs, the E-step is often the computational bottleneck because the required conditional moments are not available in closed form. Monte Carlo EM approximates these moments by simulation \citep{wei1990monte}, but accurate E-steps can require many conditional draws at each iteration. In addition, each M-step may still involve solving a nontrivial optimization problem. Recent work has developed more specialized EM-type methods for probit models. For example, \citet{ding2024computationally} propose a deterministic approximation to the required conditional moments using expectation propagation, focusing on estimation rather than inference. SEGA instead targets large-scale inference for a general class of LDVMs.

% outline.
The remainder of the paper is organized as follows. Section~\ref{sec:method} introduces the model class, the leading examples, and the SEGA estimator and variance estimator. Section~\ref{sec:application} presents two large-scale empirical applications. Section~\ref{sec:theory} establishes the finite-sample concentration and asymptotic equivalence results. Section~\ref{sec:experiments} evaluates the finite-sample performance of SEGA through Monte Carlo experiments. Section~\ref{sec:conclusion} concludes.

\section{Setup, examples and a new algorithm}\label{sec:method}
\subsection{Setting and model class}
We observe a sequence of outcomes $y_i$ and covariates $x_i$, with $d_i=(y_i,x_i)$ and $i=1,\dots,n$, generated independently from some unknown data generating process $P_0$. In LDVMs, the observed outcome variable $y_i$ represents a restricted transformation of some unobservable random variable $z_i$, which we refer to as a latent variable throughout. We write this observation rule as $y_i=\mathcal{T}(z_i)$, where $\mathcal{T}$ is known. The latent variable $z_i$ is unrestricted, but the transformation $\mathcal{T}$ induces discreteness, censoring, truncation, or other restrictions in the observed outcome. The distribution of the latent variable is governed by the unknown parameters $\theta\in\Theta\subset\mathbb{R}^{d_\theta}$.

We consider LDVMs where, given $(z_i,x_i)$, the distribution of $y_i$ is modeled using a class of probabilistic models which we represent through the following conditional latent variable representation: for each $i=1,\dots,n$, and independently across $i$,
\begin{align}\label{eq:loglikelihood}
    y_{i}&\mid z_i,x_i \;{\sim}\; p_\theta(y_{i}\mid z_i,x_i) ,\\ 
    %z_{1:n}&\mid x_{1:n} \sim p_\theta(z_{1:n}|x_{1:n}),\\
    {z_{i}}&{\mid x_{i}\;{\sim}\; p_\theta(z_{i}|x_{i}).}
\end{align} 
In many LDVMs, $p_\theta(y_i\mid z_i,x_i)$ is degenerate and simply encodes the observation rule $y_i=\mathcal{T}(z_i)$. The key feature of LDVMs is that the latent variables $z_{i}$ are required to construct the model, but are typically not themselves the object of inference. Thus, the main task in this framework is inference on $\theta$.

The key difficulty with maximum likelihood estimation in LDVMs is that the latent variables must be integrated out to conduct inference on $\theta$. The likelihood function is
\begin{equation*}
    p_\theta(y_{1:n}|x_{1:n}) = \int p_\theta(y_{1:n}|z_{1:n},x_{1:n})p_\theta(z_{1:n}|x_{1:n})\dt z_{1:n},
\end{equation*}
where $p_\theta(y_{1:n}|z_{1:n},x_{1:n}) = \prod_{i=1}^np_\theta(y_{i}|z_i,x_i)$, $p_\theta(z_{1:n}|x_{1:n}) = \prod_{i=1}^np_\theta(z_{i}|x_{i})$, and $r_{1:n} =(r_1^\top,\dots,r_n^\top)^\top$  for some random variable $r_i$. Denoting $\ell_n(\theta) = \log p_\theta(y_{1:n}|x_{1:n})$, the maximum likelihood estimator is 
\begin{equation*}
    \hat{\theta}_n = \underset{\theta\in \Theta}{\text{argmax} } \ \ell_n(\theta).
\end{equation*}
The normalized score function of $\ell_n(\theta)$ can be represented as
\begin{align}\label{Eq:Score}
    m_n(\theta):=\frac{1}{n}\nabla_\theta \ell_n(\theta)= \frac{1}{n}\nabla_\theta \log\left[ \int p_\theta(y_{1:n}|z_{1:n},x_{1:n})p_\theta(z_{1:n}|x_{1:n})dz_{1:n}\right].
\end{align}This score function depends on a high-dimensional integral, and has no analytical solution except in simple models. To account for this issue, various estimators that seek to approximate the MLE are often used when conducting inference on $\theta$. We review the most common approaches to these settings with two commonplace examples. 

\subsection{Examples and existing likelihood-based estimators}\label{sec:examples}
\subsubsection{Multinomial probit model}\label{sec:examplesMNP}
Multinomial probit models are widely used to analyze discrete choice behavior, as they allow for flexible substitution patterns among choice alternatives; see for instance, \citet{geweke2003bayesian}, \citet{natenzon2019random}, and \citet{khan2021inference}. To introduce the model, denote $y_i$ to be a multinomial choice for individual $i=1,\dots,n$, where $y_{i}=j$ if individual $i$ chooses choice alternative $j=0,1,\dots,J$. Let $z_i=(z_{i1},\dots,z_{iJ})^\top$ be a $J$-dimensional vector of continuous random variables. {The elements in $z_i$ can be interpreted as utility differences with respect to the base category $j=0$, where the differencing identifies the location of the utilities \citep{bunch1991estimability}.}

The multinomial outcome $y_{i}$ is determined by the maximum value of $z_i$:
\begin{align}\label{eq:Y_i}
y_i =
\begin{cases}
0, & \text{if } \max_{1\leq k\leq J} z_{ik}<0,\\
j, & \text{if } z_{ij}=\max_{1\leq k\leq J} z_{ik}>0.
\end{cases}
\end{align}
The latent utilities are modeled as 
\begin{align}\label{eq:Z_i}
    z_{ij} = x_{ij}^\top\beta +\varepsilon_{ij}, \quad \varepsilon_{i}=(\varepsilon_{i1},\dots,\varepsilon_{iJ})^\top\sim N(0_J,\Sigma),
\end{align}
where $x_{ij}$ is an $r$-dimensional vector, $\beta$ is an $r$-dimensional vector of coefficients,  $\varepsilon_{i}$ is a $J$-dimensional normally distributed disturbance vector with mean zero and covariance matrix $\Sigma$. We use $x_i$ to denote the collection of alternative-specific regressors, with $x_i\beta=(x_{i1}^\top\beta,\ldots,x_{iJ}^\top\beta)^\top$. The parameter vector is $\theta=(\beta^\top,\mathrm{vech}({\Sigma})^\top)^\top$.

Taking $d_i=(y_i,x_{i})$ as the observed data, the likelihood for the MNP can be stated through the choice probabilities associated with choice $j\in\{0,\dots,J\}$. Denoting $P_{ij}(\theta) = \Pr(y_i=j\mid x_i)$, we have that $P_{i0}(\theta)=\Phi_J(0_J;x_i\beta,\Sigma)$ indicates the probability of the base category, where $\Phi_J(\cdot;\mu,\Sigma)$ is the CDF of a $J$-variate normal with mean $\mu$ and covariance matrix $\Sigma$. For the remaining categories, 
\begin{flalign}
P_{ij}(\theta)&=\int \left[\mathbb{I}\left\{x_{ij}^\top \beta+\varepsilon_{ij}\ge0\right\}\prod_{k\ne j}\mathbb{I}\left\{x_{ij}^\top \beta+\varepsilon_{ij}>x_{ik}^\top \beta+\varepsilon_{ik}\right\}\right]\phi_J(\varepsilon_i;0_J,\Sigma)\dt\varepsilon_i,\label{Eq:intprobMNP}
\end{flalign}
where $\phi_J(\cdot;\mu,\Sigma)$ is the density of a $J$-variate normal with mean $\mu$ and covariance matrix $\Sigma$. If $P_{ij}(\theta)$ could be evaluated, the log-likelihood would take the simple multinomial form $\ell_n(\theta)=\sum_{i=1}^{n}\sum_{j=0}^{J}\mathbb{I}(y_i=j)\log P_{ij}(\theta)$. However, the choice probabilities \(P_{ij}(\theta)\) involve integration of a multivariate normal density over a truncated region of the latent utility space. For \(J\ge 3\), these integrals are not available in closed form \citep{botev2017normal}.

The simulated maximum likelihood (SML) estimator approximates the integrals in \eqref{Eq:intprobMNP} by simulation methods:  
$$
\hat\theta_{SML}:=\argmax_{\theta\in\Theta}\frac{1}{n}\sum_{i=1}^{n}\sum_{j=0}^{J}\mathbb{I}(y_i=j)\log \widehat{P}_{ij}(\theta),
$$
where $\wh{P}_{ij}(\theta)$ denotes the estimate for ${P}_{ij}(\theta)$ \citep{bolduc1999practical}. However, even with efficient simulation and/or approximation methods, conducting exact estimations of the MNP model in modern large-scale data sets is challenging \citep{ding2024computationally}. Indeed, the approximate nature of SML results in biased estimators of the score, with this bias often exacerbated in larger choice sets and large datasets \citep{borsch1993smooth}. An alternative to simulating the likelihood is to simulate the scores directly, ala MSS \citep{hajivassiliou1998method}. We defer a discussion of this approach to Section \ref{sec:MSS_discuss}.

\subsubsection{Random effects Tobit model}\label{sec:examplesTobit}
Random effects Tobit models are widely used to analyze censored outcomes in the presence of unobserved heterogeneity, with recent applications in marketing \citep{danaher2020advertising} and banking \citep{liu2023forecasting}. To introduce the model, let $n$ denote the total number of individuals and $T$ the total number of observations available for each individual. We denote $y_{it}\ge 0$ to be the limited dependent variable with $i =1,\dots,n$, and $t=1,\dots,T$.

The variable $y_{it}$ is modeled to be a function of the latent variable $y_{it}^*$, such that 
\begin{equation}\label{Eq:linktobit}
p(y_{it}|y_{it}^*)=\left\{\begin{matrix} \mathbb{I}(y_{it}=0) & y_{it}^*\le 0,\\
\mathbb{I}(y_{it}=y_{it}^*)&y_{it}^*>0.\\\end{matrix}\right.
\end{equation} 
This latent variable follows the normal distribution
\begin{equation}\label{Eq:yystar}
p_\theta(y_{it}^*|\alpha_i,x_{it}) = \phi_1\left(y_{it}^*;h_{it}^\top\beta+w_{it}^\top\alpha_i,\sigma^2\right),
\end{equation}
where the vector $h_{it}$ includes covariates with fixed effects $\beta$, and the vector $w_{it}$ includes covariates with random effects $\alpha_i\sim N(0_r,\Omega)$. The parameter vector is $\theta = (\beta^\top,\sigma^2,\text{vech}(\Omega)^\top)^\top$, where $\text{vech}(A)$ extracts the lower triangular elements of matrix $A$.

Taking $y_{i} = \left(y_{i1},\dots,y_{iT}\right)^\top$, $x_{i} = (x_{i1}^\top,\dots,x_{iT}^\top)^\top$, with $x_{it} = (h_{it}^\top,w_{it}^\top)^\top$, and $d_i = (y_i,x_i)$, the likelihood increment can be stated as 
\begin{equation}\label{Eq:intprobTobit}
    p_\theta(y_i|x_i) = \int g_\theta(\alpha_i,d_i)\phi_r(\alpha_i,0_r,\Omega)d\alpha_i.
\end{equation} where 
\begin{equation*}
g_\theta(\alpha_i,d_i) = \prod_{\{t:y_{it}=0\}}\Phi_1\left(0;h_{it}^\top\beta+w_{it}^\top\alpha_i,\sigma^2\right) \prod_{\{t:y_{it}>0\}}\phi_1\left(y_{it};h_{it}^\top\beta+w_{it}^\top\alpha_i,\sigma^2\right).
\end{equation*}
Because $g_\theta(\alpha_i,d_i)$ contains products of normal CDFs and must be integrated over the random effects, $p_\theta(y_i\mid x_i)$ is not available in closed form except in special cases.

To tackle this problem, Gaussian-Hermite quadrature methods can be used, which yield the approximation
\begin{equation}\label{quadrature}
\widehat{p_\theta(y_i|x_i)} =  
\frac{1}{\pi^{r/2}}
\sum_{q \in \{1,\dots,Q\}^r} 
\left(
\prod_{j=1}^{r} \omega_{q_j}
\right)
g_\theta\!\left(\sqrt{2}\,\tilde{D}\zeta(q),d_i\right),
\end{equation}
where $\Omega = \tilde{D}\tilde{D}^\top$,  $\zeta(q) = (\zeta_{q_1},\dots,\zeta_{q_r})$ and where $\{\zeta_s,\omega_s\}_{s=1}^Q$ denote the root and weight pairs of the Hermite polynomial of order $Q$ \citep{skrondal2004generalized}. A value $Q=30$ is often considered reasonable to accomplish high accuracy. However, note that the function $g_\theta(\cdot,d_i)$ must be evaluated a total $Q^r$ times, rendering quadrature methods practically infeasible beyond three dimensional random effects problems.

\subsection{Stochastically estimated gradient ascent}
As is clear from Sections \ref{sec:examplesMNP} and \ref{sec:examplesTobit}, in particular equations \eqref{Eq:intprobMNP} and \eqref{Eq:intprobTobit}, parameter estimation in LDVMs is made difficult through their dependence on the latent variables that must be integrated out before the likelihood function can be evaluated. Consequently, the likelihood score equations in \eqref{Eq:Score} depend on integrals which must generally be approximated. 

If we are willing to conduct inference via an approximation to the MLE, then we can circumvent these computational difficulties and apply LDVMs to high-dimensional settings with millions of observations and hundreds of parameters. To do so, our starting point is to re-examine the score equations in \eqref{Eq:Score}, and apply integration by parts to see that (see, e.g., \citealp{Cappe})
\begin{flalign}
    m_n(\theta)&=\frac1n\nabla_\theta 
    \log\left[\int p_\theta(y_{1:n}\mid z_{1:n},x_{1:n})p_\theta(z_{1:n}\mid x_{1:n})dz_{1:n}\right]\nonumber\\
    &=\frac1n\int_{}\nabla_\theta \log p_\theta(y_{1:n},z_{1:n}\mid x_{1:n})p_\theta(z_{1:n}\mid d_{1:n})\dt z_{1:n}.\label{Eq:FisherScore}
\end{flalign}
The gradient in \eqref{Eq:FisherScore} makes clear that if we can obtain draws from $p_\theta(z_{1:n}\mid d_{1:n})$, we could attempt to conduct inference on $\theta$ by estimating the score equations and applying modern gradient-based optimization methods. 

While we cannot evaluate $m_n(\theta)$ in general, a feasible estimator exists:
\begin{align}\label{Eq:FisherScoreest}
\widehat{m}_n(\theta)= \frac{1}{n}\nabla_\theta \log p_\theta(y_{1:n},\tilde{z}_{1:n}\mid x_{1:n}),\quad \tilde{z}_{1:n}{\sim} p_\theta(\cdot\mid d_{1:n}).
\end{align}
%While such an estimator is of course high-variance, it is nonetheless unbiased. 
Since $\wh{m}_n(\theta)$ is an unbiased estimator of $m_n(\theta)$, we use $\wh{m}_n(\theta)$ within a stochastic gradient ascent algorithm to produce stochastic iterates of $\theta$. SGA is the backbone of modern machine learning algorithms due to its ability to scale to vast amounts of data, and millions of parameters. However, our use of the estimator $\wh{m}_n(\theta)$ within an SGA algorithm is non-standard: SGA algorithms are ``stochastic'' not because they are based on gradients that depend on simulated data, but because they estimate the gradient through subsets (called mini-batches) of the data that are drawn randomly from the full data set \citep{robbins1951stochastic}. 
In contrast to standard SGA, our implementation uses the entire dataset, but is based on estimating the gradient stochastically using a subset of the possible data that could have been simulated from $p_\theta(z_{1:n}\mid d_{1:n})$. To clarify this distinction, we refer to SGA algorithms based on the estimated gradient $\wh{m}_n(\theta)$ as stochastically estimated gradient ascent (SEGA) algorithms. 

Pseudo-code for the implementation of SEGA is given in Algorithm \ref{alg:sega}. The $k$-th parameter iterate in the SEGA path is denoted as $\theta^{(k)}_n$. The iterate $\theta^{(k)}_n$ is noisy, due to stochasticity induced by the latent variables. To account for this, more precise estimates can be obtained by smoothing the iterates to obtain
\begin{align}\label{eq:smooth}
\overline\theta_n=\frac{1}{k-k_0}\sum_{j=k_0+1}^{k}\theta^{(j)}_n,
\end{align}
where $1< k_0< k$ denotes some ``burn-in'' period after which we smooth the iterates, and $k-k_0\asymp k$. This is known as Polyak–Ruppert averaging \citep{polyak1992acceleration,ruppert1988efficient}. We refer to $\overline\theta_n$ as the SEGA estimator.

\begin{algorithm}[tb!]
    \begin{algorithmic}[1]
        \State{Initialize $\theta^{(0)}_n$, $k_0$, and $k$.}
        \For{$j=0,\dots,k-1$}
        \State{Draw $\widetilde{z}^{(j)}_{1:n}{\sim} p_{\theta^{(j)}_n}(\cdot\mid d_{1:n})$.}
        \State{Calculate $\widehat{m}^{}_{n}({\theta}_n^{(j)})=\frac{1}{n}\nabla_\theta \log p_{ {\theta}_n^{(j)}}(y_{1:n},\widetilde{z}^{(j)}_{1:n})$.}
        \State{Update 
    $\theta_n^{(j+1)}= \theta_n^{(j)}+\eta^{(j+1)}_{n}\odot\widehat{m}^{}_n(\theta^{(j)}_n)$.}
        \EndFor
        \State \Return{$\overline\theta_n=\frac{1}{k-k_0}\sum_{j=k_0+1}^{k}\theta^{(j)}_n$.}
    \end{algorithmic}
    \caption{Stochastically estimated gradient ascent algorithm}
    \label{alg:sega}
\end{algorithm}

% adadelta + convergence
Step 5 in Algorithm \ref{alg:sega} requires a learning rate vector $\eta_n^{(j+1)}$. We use the ADADELTA method in \citet{zeiler2012adadelta}, which conducts exponential smoothing operations on the gradient. For $[\eta^{(j)}_{n}]_i$ denoting the $i$-th dimension of $\eta^{(j)}_{n}$, we have
\begin{equation}
    [\eta^{(j)}_{n}]_i = \frac{\sqrt{\Delta_i^{(j-1)}+\epsilon}}{\sqrt{S_i^{(j)}+\epsilon}},
\end{equation}
where $\Delta^{(j)}=\rho \Delta^{(j-1)}+(1-\rho)\left\{\eta^{(j)}_{n}\odot\widehat{m}_n(\theta^{(j)}_n)\right\}^2$ and $S^{(j)} = \rho S^{(j-1)}+(1-\rho)\widehat{m}_n(\theta^{(j)}_n)^2$,
with $\rho=0.95$, $\epsilon=10^{-6}$, and $\Delta_i^{(0)}=S_i^{(0)}=0$. This learning-rate is self-adaptive, and the magnitude of $\Delta^{(j)}$ decreases on average with $j$, which in theory allows for a stopping rule on the algorithm based on $||\Delta^{(j)}||$. In practice, such statistics are noisy, and setting a large number of iterations after which the change in the smoothed parameters is checked for convergence works best.

Step 3 of Algorithm~\ref{alg:sega} uses a single draw from the conditional distribution of the latent variables. More generally, one could generate $R>1$ draws at each iteration and replace $\widehat m_n(\theta)$ by their average. For conditionally independent draws, this estimator remains unbiased for $m_n(\theta)$, while its conditional variance decreases at rate $R$. Using multiple draws therefore produces less noisy individual updates. Increasing $R$ nevertheless raises the computational cost of each iteration, creating a trade-off between the number of draws per iteration and the number of parameter updates that can be performed for a given computational budget. 

The simulation variance of a Polyak--Ruppert average is driven by both the variance of the stochastic gradient and the number of averaged iterates. Thus, if simulation cost is approximately linear in $R$, increasing $R$ while proportionally reducing the number of iterations need not provide a first-order variance reduction relative to using a single draw and averaging over more iterations. The latter also allows the algorithm to make more parameter updates and thereby reduces the effect of initialization. We therefore use $R=1$ throughout.

\subsection{SEGA implementation in the two examples}\label{sec:implementation_examples}
SEGA requires evaluation of the augmented log gradient $\nabla_\theta \log p_\theta(y_{1:n},z_{1:n}\mid x_{1:n})$ in closed form and simulation from the distribution $p_\theta(z_{1:n}\mid d_{1:n})$. This feature allows for the direct application of SEGA in both classes of examples considered in Section \ref{sec:examples}. 

\subsubsection{SEGA for the multinomial probit model}
Since the scale of the covariance matrix of the latent utilities is not identified in a
multinomial probit model, we follow  \cite{ding2024computationally} and fix the trace of the precision matrix so that $\text{Trace}(\Sigma^{-1}) = J$, and where $\Sigma^{-1} = CC^\top$. Denote the vector $\psi = \text{vech}(C)$, which is a vector of dimension $d_C=J(J+1)/2$. The trace restriction imposes $\text{Trace}{(\Sigma^{-1})} = \sum_{l=1}^{d_C} \psi_l^2= J$. Similar to \citet{loaiza2021scalable} we parametrize $\psi$ in terms of a spherical coordinate system with radius $\sqrt{J}$  and angles $\kappa = (\kappa_1,\dots,\kappa_{d_C-1})^\top$, where
\begin{equation}\label{eq:transformation}
  \psi_{l}(\kappa) = \begin{cases}
      \sqrt{J}\cos\kappa_{1} & \text{for $l=1$},\\
    \sqrt{J}\cos\kappa_{l}\prod_{j = 1}^{l-1}\sin\kappa_{j} &   \text{for $1<l<d_C$},\\
    \sqrt{J}\prod_{j = 1}^{l-1}\sin\kappa_{j} & \text{for $l=d_C$}.
  \end{cases}
\end{equation}
To ensure positive diagonal elements of $C$, set $0<\kappa_{\frac{J(J+1)}{2}-\frac{(J-i+1)(J-i+2)}{2}+1}\le \frac{\pi}{2}$ for $i = 1,\dots,J-1$, and constrain the remaining angles to $[0,\pi]$. The angles are transformed into the real line by $\kappa_l = \Phi(\xi_l)\text{UB}_l$, where $\text{UB}_l$ denotes the upper bound of the $l^{\text{th}}$ angle. Note that $\Sigma$ is now a function of $\xi$, and $\theta = \left(\beta^\top,\xi^\top\right)^\top$.

Under this parametrization, the MNP model in \eqref{eq:Y_i} and \eqref{eq:Z_i} delivers a closed-form estimator of the score in \eqref{Eq:FisherScoreest}. The required gradient and distribution are 
\begin{flalign}
    \nabla_\theta \log p_\theta(y_{1:n},z_{1:n}\mid x_{1:n})&=\sum_{i=1}^n (\nabla_\beta \log \phi_J(z_i;x_i\beta,\Sigma)^\top,\nabla_{\xi} \log \phi_J(z_i;x_i\beta,\Sigma)^\top)^\top, \label{eq:grad_augpost_mnp}\\
p_\theta(z_{1:n}\mid d_{1:n})&=\frac{\prod_{i=1}^n\phi_J(z_i;x_i\beta,\Sigma)I[A(y_i)z_i<0]}
{\int \prod_{i=1}^n\phi_J(z_i;x_i\beta,\Sigma)I[A(y_i)z_i<0]\dt z_{1:n}}, \label{eq:pz_mnp}
\end{flalign}
Thus, $p_\theta(z_{i}\mid d_{i})$ is multivariate truncated normal with truncation matrix $A(y_i)=I_J$ if $y_i=0$ and $A(y_i)=I_J+(-e_j-1)e_j'$ otherwise, where $e_j$ is a zero vector with a one as the $j$th element. This distribution can be simulated from via Gibbs methods. Further details are provided in Supplemental Appendix~\ref{A:mnp_fisher}.

\subsubsection{SEGA for the random effects Tobit model}
In the random effects Tobit model, the latent variables $z_{1:n}$ are represented by the random effects $\alpha_i$ and the latent variables $y_{it}^\ast$. To facilitate optimization, all parameters are transformed to the real line. The variance parameter is transformed as $c = \log(\sigma^2)$. The covariance matrix of the random effects is written down in terms of its Cholesky decomposition $\Omega = \tilde{D}\tilde{D}^\top$, where $\tilde{D}$ is a lower triangular matrix with diagonal elements $\tilde{d}_{i,i} = \exp(\delta_{i,i})$, and off diagonal elements $\tilde{d}_{i,j} = \delta_{i,j}$. 
Denote $D$ to be the lower triangular matrix comprised of the elements $\delta_{i,j}$, and define
$\delta= \text{vech}(D)$ to be the vector of unique parameter vectors that characterize $\Omega$. The parameter vector is $\theta = (\beta^\top,c,\delta^\top)^\top$.

The random effects Tobit model in \eqref{Eq:linktobit} and \eqref{Eq:yystar} also delivers the gradient and distribution required by the estimator of the score  in~\eqref{Eq:FisherScoreest}:
\begin{flalign}\label{eq:tobit_sega}
    \nabla_\theta \log p_\theta(y_{1:n},z_{1:n}|x_{1:n})&=\sum_{i=1}^n\left(\sum_{t=1}^T\nabla_{\beta,c}\log\phi_1(y_{it}^*;\mu_{it},\sigma^2)^\top,\nabla_{\delta}\log\phi_r(\alpha_i;0_r,\Omega)^\top\right)^\top,\notag \\
p_\theta(y_{1:n}^*\mid d_{1:n},\alpha_{1:n})&=\prod_{\{(i,t):y_{it}=0\}}\frac{\phi_1(y_{it}^*;\mu_{it},\sigma^2)I(y_{it}^*<0)}
{\Phi_1(0;\mu_{it},\sigma^2)}\prod_{\{(i,t):y_{it}>0\}} I(y_{it}^*=y_{it}),\notag \\
p_\theta(\alpha_{1:n}\mid d_{1:n},y^*_{1:n})&=\prod_{i=1}^n\phi_r(\alpha_i;\bar{\alpha}_i,V_{i}),
\end{flalign}
\sloppy with $\mu_{it}=h_{it}^\top\beta+w_{it}^\top\alpha_i$, $\bar{{\alpha}}_i =\frac{1}{\sigma^2}V_i\left(\sum_{t=1}^{T}w_{it}{y}_{it}^*-\left[\sum_{t=1}^{T}w_{it}h_{it}^\top\right]\beta\right)$, and $V_i = \left[\Omega^{-1}+\frac{1}{\sigma^2}\sum_{t=1}^{T} w_{it}w_{it}^\top\right]^{-1}$. The latent variables $y^*_{1:n}$ can be simulated from a univariate truncated normal and $\alpha_{1:n}$ from a multivariate normal via Gibbs methods. Further details are provided in Supplemental Appendix~\ref{A:tobit_fisher}.

\subsection{Inference after SEGA}\label{sec:inference}
Although the SEGA estimator is obtained from stochastic gradients rather than direct likelihood maximization, in Section~\ref{sec:theory} we show that $\bar\theta_n$ in \eqref{eq:smooth} is asymptotically equivalent to the infeasible MLE under regularity conditions. Hence, inference can be based on the same sandwich covariance matrix as the infeasible MLE; namely, $V=\mathcal{H}^{-1}\mathcal{I}\mathcal{H}^{-1}$, where $\mathcal{H}=-\E\nabla_{\theta}^2\log p_\theta (d_i)$ is the expected Hessian of the log likelihood and $\mathcal{I}=\text{Var}\{\nabla_\theta \log p_\theta (d_i)\}$ the variance of the score.

\subsubsection{A variance estimator for SEGA}
Estimating $V$ requires estimates of $\mathcal{H}$ and $\mathcal{I}$, which are not directly available because the observed likelihood score and Hessian involve integrals over latent variables. We therefore estimate these quantities using Louis' identity \citep{louis1982finding} and Fisher's identity, respectively, where conditional expectations are approximated by simulation.  

Given simulated data, $z_{i,s}\stackrel{iid}{\sim} p_\theta(z_i|d_i)$, for $\ell(\theta
,z_i) = \log p_\theta(y_i,z_i|x_i)$,  and $\widehat{\nabla_\theta\ell_i(\theta)} = \frac{1}{S}\sum_{s=1}^S\nabla_\theta\ell(\theta,z_{i,s})$, our proposed estimator is 
\begin{align}\label{eq:Vhat}
    \widehat{V} = \widehat{\mathcal{H}(\bar{\theta}_n)}^{-1}
\widehat{\mathcal{I}(\bar{\theta}_n)}
\widehat{\mathcal{H}(\bar{\theta}_n)}^{-1},
\end{align}
where 
\begin{align*}
\widehat{\mathcal{I}(\theta)} =&  \frac{1}{n}\sum_{i=1}^n\widehat{\nabla_\theta\ell_i(\theta)}\widehat{\nabla_\theta\ell_i(\theta)}^\top-\left(\frac{1}{n}\sum_{i=1}^n\widehat{\nabla_\theta\ell_i(\theta)}\right)\left(\frac{1}{n}\sum_{i=1}^n\widehat{\nabla_\theta\ell_i(\theta)}\right)^\top,\\
    \widehat{\mathcal{H}(\theta)} =& \frac{1}{n}\sum_{i=1}^n \widehat{\nabla_\theta\ell_i(\theta)}\widehat{\nabla_\theta\ell_i(\theta)}^\top   -\frac{1}{S}\sum_{s=1}^S\left[\nabla_\theta^2 \ell(\theta,z_{i,s})+
      \nabla_\theta \ell(\theta,z_{i,s})\nabla_\theta \ell(\theta,z_{i,s})^\top\right].
\end{align*}
Each of the terms can be estimated separately using the simulated data $z_{i,s}$. Fisher's identity allows the observed score for each unit to be recovered by averaging the augmented score $\nabla_\theta \log p_\theta(y_i,z_i|x_i)$ over independent draws from the conditional distribution of the latent variables $p_\theta(z_i|d_i)$. The estimator $\widehat{\mathcal{I}(\theta)}$ is therefore the sample variance of these estimated observed scores. To estimate $\mathcal{H}$, we use Louis' identity, which  expresses the observed Hessian as the conditional mean of the augmented Hessian $\nabla^2_\theta \log p_\theta(y_i,z_i|x_i)$ plus a correction for the conditional variation of the augmented score. 

Tests and confidence intervals can then be based on $\widehat V$, while for nonlinear economic quantities, such as elasticities and marginal effects, we use a \citet{krinsky1986approximating} parametric bootstrap: we draw parameters from the estimated asymptotic distribution of $\overline\theta_n$ and transform each draw into the corresponding quantity of interest.

\subsubsection{Evaluating the variance estimator in the SEGA examples}
The evaluation of $\widehat V$ requires $S$ independent draws from $p_\theta(z_i|d_i)$. Since the Gibbs samplers used in Section~\ref{sec:implementation_examples} for SEGA point estimation produce autocorrelated draws, they cannot be used for estimation of $\widehat{V}$ without modifications.
For the multinomial probit model, we use the accept-reject sampling method proposed by \citet{botev2017normal} for generating i.i.d.\ draws from the truncated multivariate normals implied by $p_\theta(z_{1:n}\mid d_{1:n})$. This approach is computationally more demanding than the Gibbs sampler used in SEGA, and is therefore solely for the construction of $\widehat V$.
In the random effects Tobit model, we employ a two step generation process based on $p_\theta(z_i|d_i) = p_\theta(y_i^*|d_i)p_\theta(\alpha_i|d_i,y_i^*)$. The density $p_\theta(y_i^*|d_i)$ can be written as a truncated multivariate normal, for which we can again use the \citet{botev2017normal} method for generating i.i.d.\ draws. The density $p_\theta(\alpha_i|d_i,y_i^*)$ is multivariate normal and also part of the Gibbs sampling scheme in SEGA. 

Given the simulated latent variables, both models reduce to Gaussian augmented likelihoods, so the augmented score $\nabla_\theta \log p_\theta(y_i,z_i|x_i)$ and augmented Hessian $\nabla^2_\theta \log p_\theta(y_i,z_i|x_i)$ entering Fisher's and Louis' identities are available in closed form. Details are provided in Supplemental Appendix~\ref{A:MNPse} and~\ref{A:Tobitse}.

\subsection{SEGA and the method of simulated scores}\label{sec:MSS_discuss}
In contrast to standard simulation-based estimators, SEGA attempts to solve the score equation $m_n(\theta) = 0$ stochastically, rather than deterministically. Thus, SEGA iterates will not satisfy the equation $m_n(\theta^{(k)}_n)=0$ for any finite $k$, and in general neither will $\overline\theta_n$. As such, SEGA iterates cannot be interpreted as Z-estimators (\citealp{van2000asymptotic}, Ch. 5) or simulated method of moments estimators \citep{pakes1989simulation}. 

Nevertheless, SEGA estimators share a similar philosophy to simulation-based econometric estimators, such as the MSS estimators presented in \cite{hajivassiliou1996simulation}. {Indeed, the replacement of the intractable score $m_n(\theta)$ with an average over a large number of simulated draws is the key idea in MSS. In particular,} the MSS estimator is also based on an approximation to the scores of LDVMs, and can be thought of as solving a simulated version of the scores
\begin{equation}
\widehat\theta_{\mathrm{MSS}}^{(R)}\in\left\{\theta\in\Theta:\frac{1}{n}\sum_{i=1}^n\frac{1}{R}\sum_{s=1}^R\widetilde m_i\{\theta,z_i(\theta;\epsilon_{i,s})\}=0\right\},\label{eq:MSS}
\end{equation}
where $z_i(\theta;\epsilon_{i,s})$ is drawn from $p_\theta(z_{i}\mid d_{i})$ based on the fixed set of random numbers $\epsilon_{i,s}$, for $i\in\{1,\dots,n\}$ and  $s\in\{1,\dots,R\}$, that do not change across the algorithm, and where $\tilde{m}_i\{\theta,z_{i}(\theta;\epsilon_{i,s})\}$ denotes a per-unit estimator of the corresponding per-unit infeasible score equations $m_i(\theta)$, where $m_n(\theta)=\frac{1}{n}\sum_{i=1}^{n}m_i(\theta)$. 

For a given $\theta$ value, implementation of MSS requires generating $z_i(\theta;\epsilon_{i,s})$ using a set of fixed random numbers $\epsilon_{i,s}$  to ensure that $z_i(\theta;\epsilon_{i,s})$ cannot change in a stochastic manner. This requirement is necessary for consistent estimation in MSS. As can be seen by \eqref{eq:MSS}, the MSS estimator is a Z-estimator based on numerically solving an approximation to the infeasible score equations $0=m_n(\theta)$. Consequently, as has been known since at least \cite{pakes1989simulation}, if MSS allows the random draws $\epsilon_{i,s}$ to change across evaluations of the criterion in \eqref{eq:MSS}, the MSS estimator $\widehat\theta_{\mathrm{MSS}}^{(R)}$ will not be consistent. 

However, simulating draws $z_i(\theta;\epsilon_{i,s})$ that only depend on the fixed set of random numbers $\epsilon_{i,s}$ is no simple task in LDVMs. In general, such an approach is infeasible in any setting where generating random variables requires accept/reject sampling or other techniques that do not allow for the random numbers to be fixed in their implementation. %While Gibbs sampling can be used to generate $z_i(\theta;\epsilon_{i,s})$, the resulting parameter estimates can be highly-sensitive to the number of Gibbs steps used to generate each $z_i(\theta;\epsilon_{i,s})$ as well as the initialization of the Gibbs chain. 
In cases where Gibbs sampling is used to simulate draws of $z_i(\theta;\epsilon_{i,s})$, the accuracy of the resulting parameter estimates depend on the interaction between the sample size, $n$, the number of simulations used to estimate the gradient, $R$, and the number of steps used in the Gibbs sampler to draw each $z_i(\theta;\epsilon_{i,s})$, as well as its initialization. 
In our numerical experiments conducted in Section \ref{sec:MSS_Emp}, we demonstrate that even in moderate dimensions the MSS estimator requires a large number of simulations $R$ and a large number of steps within the Gibbs sampler to accurately approximate the MLE, which becomes computationally prohibitive even in moderate dimensions. 

The requirement that the latent variables be simulated as $z_i(\theta;\epsilon_{i,s})$ with $\epsilon_{i,s}$ fixed across evaluations, is absent in SEGA: each $\widetilde{z}^{(k)}_{i}\stackrel{iid}{\sim} p_{\theta^{(k)}_n}(\cdot\mid d_{i})$
is generated using new random numbers at each iteration. SEGA bypasses this requirement because it only requires the stochastic iteration path to converge to a solution of $m_n(\theta)=0$. 

\section{Empirical applications}\label{sec:application}
To illustrate the empirical relevance of SEGA, we apply the method to two large-scale LDVMs for which likelihood-based inference is infeasible with existing approaches: a multinomial probit model for pasta brand choice and a random effects Tobit model for household-level pasta demand. We implement SEGA with 100,000 iterations, use the final 25{,}000 iterations to construct the smoothed estimates, and use $S=20{,}000$ draws to construct the variance estimator. Details are deferred to Supplemental Appendix~\ref{A:mnp_application} and \ref{A:tobit_application}.

\subsection{Brand choice and substitution in the pasta market}\label{sec:app_mnp}
We apply the MNP model to a data set on pasta brand purchases using SEGA. This consumer choice data set, collected at a leading U.S. grocer, includes more than one million purchases and is made available by the Dunnhumby data platform\footnote{https://www.dunnhumby.com/source-files/} as ``Carbo-Loading: A Relational Database". The final sample consists of purchases without coupons of the ten top-selling pasta brands, excluding private labels, amounting to 1,070,436 observations. We include an intercept and the log price for each brand in the model. \citet{loaiza2023fast} and \citet{loaiza2024hybrid} fitted this MNP model to this data using approximate Bayesian estimation methods. 
Application of the MNP model to data problems of this magnitude -- large choice set and more than a million observations -- with the theoretical guarantees associated with maximum likelihood estimation is only feasible through the use of SEGA. 

Due to the differenced utility representation in the MNP model, the model parameters do not provide direct economic interpretation. Instead, empirical interest typically focuses on the price elasticities and the substitution patterns implied by the model parameters.

Figure~\ref{fig:ped} shows the estimated own-price elasticity of the pasta brand `Barilla' in Panel (a) and the cross-price elasticity of `San Giorgio' with respect to the price of `Barilla' in Panel (b). The prices of the other pasta brands are fixed at their mean. The solid lines represent the point estimates for the price elasticities, and the dashed lines the 95\% confidence intervals. As expected, the estimated own-price elasticity is negative, and the confidence intervals are narrow due to the large sample size. The positive cross-price elasticity indicates that increases in the price of Barilla raise the probability of purchasing San Giorgio, consistent with the two brands being substitutes.

\begin{figure}[tb!]
\caption{Price elasticities for two pasta brands}
\centering
\includegraphics*[width=\textwidth]{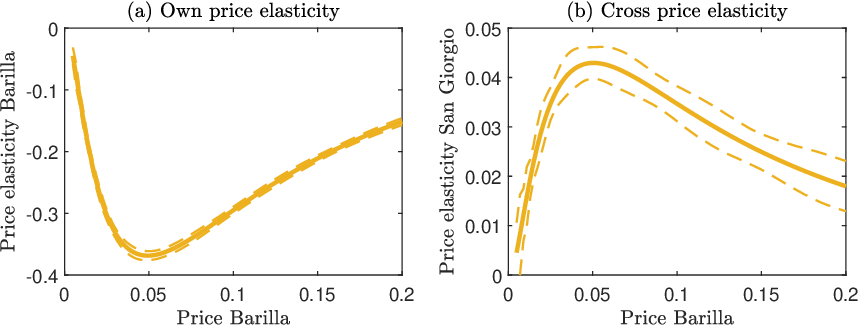}
\begin{flushleft}
\footnotesize
This figure shows the price elasticity of the pasta brands `Barilla' (Panel (a)) and `San Giorgio' (Panel (b)) as a function of the price of Barilla, with the prices of the other pasta brands fixed at their mean. The solid lines represent the price elasticities and the dashed lines the 95\% confidence intervals.  
\end{flushleft}
\label{fig:ped}
\end{figure}

%Discuss wald tests. 
Empirical applications of the MNP model often impose covariance restrictions. \citet{paetz2018utility} discuss the implied substitution patterns of an identity covariance matrix, \citet{rossi2012bayesian} justify a diagonal matrix by arguing that choice behavior is often characterized by large differences in relative variance between choice alternatives, and \citet{geweke1994alternative} show that assuming an identity covariance matrix for the undifferenced utilities implies an equicorrelated matrix on the differenced utilities. SEGA enables formal testing of these specifications via Wald tests. The restrictions $H_0:\Sigma=I_J$, $H_0:\sigma_{jk}=0$ for all $j\neq k$, and $H_0:\Sigma=\frac{1}{2}(I_J+\iota_J\iota_J^\top)$ are all decisively rejected, with test statistics of $27{,}322$, $24{,}094$, and $9{,}451.6$, respectively, far exceeding their corresponding 1\% critical values of $68.709$, $58.619$, and $68.709$. These results provide strong evidence against commonly imposed covariance restrictions and indicate that allowing unrestricted dependence in latent utilities is empirically consequential for substitution patterns.

\subsection{Household heterogeneity in pasta demand}\label{sec:app_tobit}
We fit the random effects Tobit model to a dataset constructed from the same consumer choice data used in Section~\ref{sec:app_mnp}. We select the purchases of private label pasta products that do not involve coupons from $n=203,965$ households across $T=62$ weeks. We model in each week $t$ the total dollar sales per ounce per household $i$, conditional on a set of promotion indicators, a price index, lagged sales, and weekly dummies. The 16 promotion variables indicate the number of shopping trips during the week with a private label pasta product in the weekly mailer or in a temporary in-store display, with different indicators for different locations in the mailer or display. The standardized log price index is constructed as the arithmetic mean across the price of all products in a particular week, where the weekly price for each product is approximated as the average price over all transactions for that product in that week. The lagged dollar sales for each household in each week is measured as an average across the past four weeks. The model includes a random intercept and random coefficients for the price index and lagged sales.  

Figure~\ref{fig:tobit} shows the estimated probability of positive pasta sales in Panel (a) and the estimated marginal effect of the log price index on expected sales in Panel (b), both as a function of the household-specific lag-sales coefficient. Promotions are set to zero, and Panel (a) conditions on the average price index and zero lagged sales, and Panel (b) on the median price index across observations with positive sales and the median positive lagged sales. Panel (a) shows that households with stronger lag-sales effects have lower baseline purchase inclination. Panel (b) shows that the price marginal effect becomes more negative for households with stronger lag-sales effects, while the uncertainty band widens at high lag-sales coefficients. Together, these patterns suggest that households with low baseline demand are more state dependent and, when they have recently purchased, exhibit stronger price responsiveness.

\begin{figure}[tb!]
\caption{Household heterogeneity and state dependence in pasta demand}
\centering
\includegraphics*[width=\textwidth]{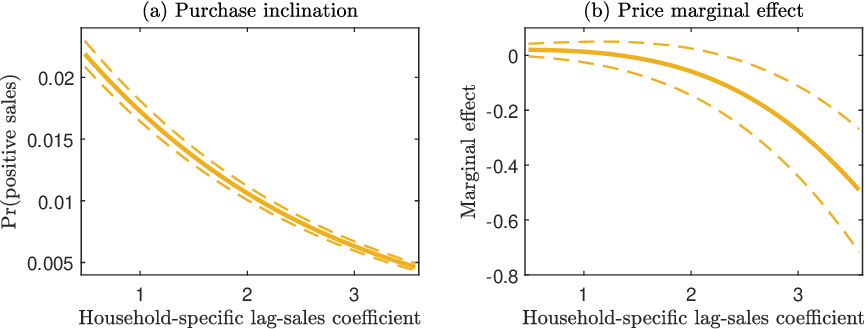}
\begin{flushleft}
\footnotesize
This figure shows the probability of positive pasta sales with average price index and zero lagged sales (Panel (a)) and the marginal effect of the log price index on expected sales with median price index and median positive lagged sales (Panel (b)), as a function of the household-specific lag-sales coefficient. Promotions are set to zero. The solid lines represent the point estimates and the dashed lines the 95\% confidence intervals. 
\end{flushleft}
\label{fig:tobit}
\end{figure}

% test offdiags covmat re are zero
% specifically test corr intercept and lag coef
The random effects Tobit model allows for household heterogeneity in baseline purchase inclination, price sensitivity, and state dependence. Restricted specifications often assume that these sources of heterogeneity are independent, or that state dependence is homogeneous across households. SEGA enables formal testing of such restrictions via Wald tests. We first test whether the random effects are independent, $H_0:\Omega_{0p}=\Omega_{0\ell}=\Omega_{p\ell}=0$, where the subscripts denote the random intercept, price-index coefficient, and lag-sales coefficient. We also test the restriction that baseline purchase inclination and state dependence are unrelated, $H_0:\Omega_{0\ell}=0$. Both restrictions are rejected, with both test statistics larger than 2900, exceeding their corresponding 1\% critical values of 11.345 and 6.635. % 2.9730e+03 and 2.9029e+03
\section{Theoretical properties of SEGA}\label{sec:theory}
In this section, we show that SEGA produces inferences that are first-order asymptotically equivalent to the infeasible MLE when the number of iterations in the algorithm scale appropriately with $n$. Critically, we do not require that the number of simulations used to approximate the score diverges: first-order asymptotic equivalence follows using a single simulated draw of the latent variables at each iteration of the algorithm. 

\subsection{Finite-sample concentration}
In this section, we analyze the randomness in SEGA due to the simulated latent variables generated along the iteration path. To make the algorithm's dependence on the simulated variables $z_{1:n}$ explicit, in what follows we write $\wh{m}_n(\theta;z_{1:n})$. 

The iterates in Algorithm~\ref{alg:sega} can be written as
\begin{align}
\theta^{(k+1)}_n&=\theta^{(k)}_n+\eta^{(k+1)}_nm_n(\theta^{(k)}_n)-\eta_n^{(k+1)}\left\{m_n(\theta^{(k)}_n)-\widehat{m}_n(\theta^{(k)}_n,\widetilde{z}^{(k)}_{1:n})\right\},
\label{eq:new}
\end{align}
where the notation $\widetilde{z}^{(k)}_{1:n}$ clarifies that the latent variables used to estimate the score are drawn at the $k$-th step of the algorithm, and are therefore conditional on $\theta_n^{(k)}$. The first term in \eqref{eq:new} is the infeasible gradient-ascent update based on the exact normalized score $m_n(\theta)$. The second term can be viewed as simulation noise: under Fisher's identity, this noise is a martingale difference sequence (MDS) conditional on the past iterates for any $k\ge1$. Hence, the theoretical behavior of $\theta_n^{(k+1)}$ will be determined by how we can control the MDS term. The formulation in \eqref{eq:new} also shows that the randomness in $\theta_n^{(k+1)}$ is due to the conditional simulation of $\widetilde{z}_{1:n}^{(k)}$ - generated conditionally on $\theta_n^{(k)}$ and the observed data $d_{1:n}$ - clarifying that this randomness depends on the observed path of iterates $\{\theta_n^{(0)},\theta_n^{(1)},\dots,\theta_n^{(k)}\}$.

These observations imply that SEGA estimators have different theoretical behavior to that of other simulation-based estimators, such as MSS. In particular, the conditional simulation of $\widetilde{z}_{1:n}^{(k)}$ introduces feedback between the simulated datasets and the updated parameters values: a draw of $\widetilde{z}_{1:n}^{(k)}$ in the tail of the distribution $p_\theta(\widetilde{z}_{1:n}^{}\mid d_{1:n})$ could deliver an iterate $\theta_n^{(k+1)}$ that lies quite far away from the MLE. Hence, the analysis of any individual parameter iterate may not be that informative, as such draws of $\widetilde{z}_{1:n}^{}$ will occur in practice. Consequently, the most useful way to understand the behavior of the sequence $\theta_n^{(k+1)}$ is by analyzing its average behavior, where we average over $\widetilde{z}_{1:n}$ and $d_{1:n}$. 

To investigate the average behavior of $\theta_n^{(k)}$, we maintain a set of regularity conditions for the conditional expectation of  $\widehat{m}_n(\theta;\widetilde{z}_{1:n})$ that are similar to those encountered in the literature on SGA methods. % (see, e.g., \citealp{moulines2011non}).
Let $\E$ denote the expectation with respect to the observed and simulated data. Write $\wh{m}_n(\theta^{(k)}_n; z^{(k)}_{1:n})=\frac{1}{n}\sum_{i=1}^{n}g_i(\theta^{(k)}_n;z^{(k)}_i)$, where $g_i(\theta,z_i)=\nabla_\theta\log p_\theta(y_i,z_i\mid x_i)$ denotes the augmented-data score contribution of observation $i$. Let $(\mathcal F_k)_{k\ge0}$ denote the sigma-field generated by the observed data, the initial value $\theta_n^{(0)}$, and all simulated latent variables used up to iteration $k-1$. We maintain the following assumptions.

\begin{assumption}\label{ass:mds}
%Let $(\mathcal{F}_k)_{k\ge0}$ be an increasing family of $\sigma$-fields and let  $\theta^{(0)}_n$ be $\mathcal{F}_0$-measurable. 
The following conditions hold with probability one (wp1). 

\noindent (i) For each $k\ge1$:
$$
\E[\widehat{m}_n(\theta^{(k)}_n;z^{(k)}_{1:n})\mid \mathcal{F}_{k}]=m_n(\theta^{(k)}_n);
$$ 
\noindent(ii) There exists a $\nu\in\mathbb{R}_+$ such that, for all $i=1,\dots,n$, and for each $k\ge1$: 
$$\E\left[\|g_i(\theta^{(k)}_n;z^{(k)}_{i})-\E[g_i(\theta^{(k)}_n;z^{(k)}_{i})\mid\mathcal{F}_{k}]\|_2^2\mid\mathcal{F}_{k}\right]\le\nu^2.
$$
\end{assumption}

Assumption~\ref{ass:mds} collects the regularity conditions required for the simulated score. Part~(i) is the conditional unbiasedness implied by Fisher's identity. This is automatically satisfied by SEGA if $z^{(k)}_{1:n}$ is an exact draw from its conditional distribution. If draws are produced by Gibbs sampling, this condition requires the Gibbs chain to be at stationarity. The numerical experiments in Section~\ref{sec:experiments} suggest that the potential error from taking a single Gibbs step at each iteration of the algorithm is negligible. Part~(ii) requires the stochastic score to have a conditional variance that is uniformly bounded along the iterate path.

\begin{assumption}\label{ass:conv} The objective function $\ell_n(\theta)$ satisfies the following. 

\noindent (i) There exists a $\mu>0$ such that, for any $\theta_2,\theta_1\in\Theta$,
$$ \ell_n\left(\theta_2\right) \leq \ell_n\left(\theta_1\right)+\left\langle\nabla_\theta \ell_n\left(\theta_1\right), \theta_2-\theta_1\right\rangle-\frac{n\mu}{2}\left\|\theta_1-\theta_2\right\|_2^2.$$ 
\noindent (ii) There exists an $L>0$ such that, for any $\theta\in\Theta$,  $$\left\|\frac{1}{n}\{\nabla_\theta\ell_n(\theta)-\nabla_\theta\ell_n(\hat\theta_n)\}\right\|\le L\|\theta-\hat\theta_n\|.$$%Further, assume that $\nabla_{\theta\theta}^2\ell_n(\hat\theta)$ exists.
\end{assumption}

Assumption~\ref{ass:conv}(i) is a strong-concavity condition on the sample log likelihood; this condition is equivalent to the requirement that the normalized objective $n^{-1}\ell_n(\theta)$ is $\mu$-strongly concave, so the normalized score $m_n(\theta)=n^{-1}\nabla_\theta\ell_n(\theta)$ is strongly monotone with curvature $\mu$. This condition is stronger than what is required for asymptotic likelihood theory, where local curvature around the MLE is sufficient. We impose it here to obtain a finite-$n$ and finite-$k$ concentration bound for $\overline\theta_n$. An asymptotic version can be obtained under weaker local conditions. %Assumption~\ref{ass:conv}(ii) is a Lipschitz condition on the exact score, and can be shown to be satisfied under smoothness conditions and restrictions on the parameter space. 
Assumption~\ref{ass:conv}(ii) bounds changes in the normalized observed-data score relative to the MLE. This is weaker than global Lipschitz continuity and is satisfied, for example, when the normalized likelihood Hessian is uniformly bounded over the parameter space.
With effort, Assumptions \ref{ass:mds}-\ref{ass:conv} can be verified in the MNP example under low-level sufficient conditions. We refer the interested reader to Supplemental Appendix~\ref{A:lemmas} for details.

%While Assumptions~\ref{ass:mds}-\ref{ass:conv} may appear strict, 
Related conditions to the ones in Assumptions~\ref{ass:mds}-\ref{ass:conv} have been maintained by several authors when providing non-asymptotic convergence results for SGA algorithms that are different to SEGA.  Similar versions of Assumption~\ref{ass:mds}(i) were maintained by both \cite{moulines2011non} and \cite{chen2020statistical}, and Assumption \ref{ass:conv} is a slightly weakened version of Assumption 3.1 in \cite{chen2020statistical}. However, Assumption~\ref{ass:mds}(ii) is a much weaker requirement than what is maintained in the SGA literature. This weaker condition can be maintained since the randomness in SEGA is generated from conditionally independent simulated random variables, and not caused by resampling the observed data. As we show in Lemma \ref{thm:main} in Appendix~\ref{A:proofs}, this allows us to derive a faster rate of convergence for the SEGA per-iteration error, and ultimately allows us to choose a learning rate that is much larger than what is generally maintained in pure Robbins-Monro schemes. 

\begin{assumption}\label{ass:iters}
{The step-size satisfies $\eta_n^{(k)}\propto\eta_n k^{-\alpha}$, with $\alpha\in(0,1)$,} and $0<\underline{c}\le \eta_n\le \overline{c}$ (wp1), with $\eta_n$ being $\mathcal{F}_0$ measurable.
\end{assumption}
The condition $\alpha\in(0,1)$ ensures that the step size converges to zero while the cumulative step size diverges, allowing the algorithm to eliminate the effect of its initial value. This condition is sufficient for the $L^p$ $(p\in\{1,2\})$ convergence results below and is weaker than the standard Robbins--Monro condition of $\alpha>1/2$ used to establish almost-sure convergence. Although the original formulation of ADADELTA  given in Algorithm \ref{alg:sega} performs well in practice, its update need not satisfy Assumption \ref{ass:iters}. One may therefore use ADADELTA during the initial iterations and, once the smoothed iterates plateau, switch to the scaled update in Assumption \ref{ass:iters}.

We now state our first result, which bounds the deviation between the MLE and the SEGA estimator. 
Proofs of all subsequent results are deferred to Appendix~\ref{A:proofs}. 
\begin{theorem}\label{corr:smoothed}
Under Assumptions \ref{ass:mds}-\ref{ass:iters},
$$\E\|\overline\theta_n-\hat\theta\|^p\lesssim 
%\left\{\E\|\theta^{(0)}_n-\hat\theta_n\|^2_2\exp(-  k^{1-\alpha})\right\}^{p/2} +\left({nk^{\alpha}}\right)^{-p/2},\quad p\in\{1,2\}.
k^{-p}\left\{\E\|\theta^{(0)}_n-\hat\theta_n\|^2_2\right\}^{p/2}+(nk^\alpha)^{-p/2},\quad p\in\{1,2\}.$$
\end{theorem}

Theorem~\ref{corr:smoothed} separates the approximation error of the SEGA estimator into an initialization and a simulation component. The former is of order $k^{-1}$, while the latter is of order $(nk^\alpha)^{-1/2}$. Since each SEGA update averages simulation noise across the full sample, the simulation component includes the factor $n^{-1/2}$. For any fixed sample size, both components vanish as the number of iterations increases. 

\subsection{Asymptotic equivalence to the infeasible MLE}
We now study the asymptotic behavior of $\bar\theta_n$ as both the sample size \(n\) and the number of iterations \(k_n=k(n)\) diverge.
%
%\subsubsection{Consistency}
The result in Theorem \ref{corr:smoothed} states that the average deviation between $\overline\theta_n$ and $\hat\theta_n$ can be bounded by a constant multiple of $(nk^\alpha)^{-1/2}$ and the difference between the starting value of the algorithm, $\theta^{(0)}_n$, and the MLE $\hat\theta_n$. %So long as $\E\|\theta_n^{(0)}-\hat\theta_n\|\le C_n$, which includes the case where $\theta_n^{(0)}$ is an arbitrary non-random starting point restricted to lie in $\Theta$, with $\Theta$ compact, then we can immediately prove convergence in distribution of $\overline{\theta}^{}_n$. To state such a result, we make the relationship between $k$ and $n$ formal via the following assumption.
We first control the effect of the initial value.

\begin{assumption}\label{ass:starting}
There exists a random variable $C_n$ such that $\|\theta_n^{(0)}-\hat{\theta}_n\| \leq C_n$ where $\E(C_n^2)<\infty$.
\end{assumption}

Assumption~\ref{ass:starting} requires the distance between the initial value and the MLE to remain bounded as the sample size grows. This condition is automatically satisfied when $\Theta$ is compact and both $\theta_n^{(0)}$ and
$\hat{\theta}_n$ lie in $\Theta$. Importantly, Assumption~\ref{ass:starting} does not require the initial value to be consistent for the population parameter or the infeasible MLE. Although such an assumption is common in the analysis of SGA algorithms, it would limit the applicability of SEGA since an initially consistent estimator is difficult to obtain in practice for large LDVMs.

Under Assumption~\ref{ass:starting}, Theorem~\ref{corr:smoothed} implies $\sqrt{n}\E\left\|\bar{\theta}_n-\hat{\theta}_n\right\|\lesssim {\sqrt{n}}k_n^{-1}+{k_n^{-\alpha/2}}$.
The first term captures the remaining effect of initialization in the averaged iteration path, whereas the second captures simulation noise. The simulation component vanishes for any $k_n\to\infty$ but the impact of the averaged iteration path depends on the initialization of the algorithm. If we have that $\mathbb{E}\|\sqrt{n}(\theta^{(0)}_n-\hat\theta_n)\|=O(1)$, then any $k_n\to\infty$ will ensure that $\overline\theta_n$ is asymptotically first-order equivalent to the MLE. For example, $k_n\gg\log(n)$ suffices. In contrast, allowing an arbitrary initial value requires the number of iterations to grow sufficiently fast with the sample size. This motivates the following condition.

\begin{assumption}\label{ass:iters2}
 %The number of iterations, $k_n=k(n)$, satisfies $k_n\rightarrow\infty$ as $n\rightarrow\infty$. 
{The number of iterations, $k_n=k(n)$, satisfies $k_n\gg \sqrt{n}$ as $n \rightarrow \infty$.}
\end{assumption}

The following result shows that the SEGA estimator $\overline\theta_n$ is first-order asymptotically equivalent to the infeasible MLE.

\begin{theorem}\label{thm:limit}
 Assumptions \ref{ass:mds}-\ref{ass:iters2} are satisfied. If $\sqrt{n}(\hat{\theta}_n-\theta_0)\Rightarrow N(0,V)$,  then  $\sqrt{n}(\overline\theta_n-\theta_0)\Rightarrow N(0,V)$.
\end{theorem}
Theorem \ref{thm:limit} is the main inferential result for SEGA: for a sufficiently large number of iterations, SEGA inherits the limiting distribution of the infeasible MLE, including the usual sandwich form under misspecification as defined in Section~\ref{sec:inference}.

\section{Numerical experiments}\label{sec:experiments}
We use numerical experiments to evaluate the finite-sample performance of SEGA in the multinomial probit model. We compare SEGA with the method of simulated scores (MSS), simulated maximum likelihood (SML), and the exact MLE. We examine whether the theoretical concentration and asymptotic equivalence results for SEGA translate into accurate and computationally scalable finite-sample inference. Experiments with the random effects Tobit model are reported in Supplemental Appendix~\ref{A:tobit_simulation}.

\subsection{Design and implementation}
We consider the MNP model described in Section~\ref{sec:examplesMNP}. The data is generated according to \eqref{eq:Y_i} and \eqref{eq:Z_i}. The predictor matrix $x_i$ has rows $x_{ij} = (e_j^\top, \Delta p_{ij})$, where $\Delta p_{ij} = \log p_{ij} - \log p_{i0}\sim N(0,1)$ denotes the log price difference of alternative $j$ relative to the base. Hence, the model contains $J$ alternative-specific intercepts and one price coefficient. All experiments use $\Sigma=J \tilde{\Sigma}/\text{trace}(\tilde{\Sigma})$ with $\tilde{\Sigma}= I_J +\alpha \alpha^\top$ and  $\alpha_j=1+0.5j/J$.  Experiments with $J$ equal to 2, 3, and 6, set $\beta = [0.5, -0.5, -0.8]^\top$, $\beta = [0.5, -0.5, -0.2, -0.8]^\top$, and $\beta = [-0.5, -0.30,-0.1,0.1,0.3,0.5,-0.8]^\top$, respectively. The trace normalization keeps the scale comparable across choice alternatives, but we conduct inference after scaling to $\Sigma_{11}=1$ to avoid a joint restriction across all covariance parameters. We therefore report estimates for $\beta/\sqrt{\Sigma_{11}}$ and $\Sigma/\Sigma_{11}$.

SEGA is implemented with 20,000 iterations with the final 1,000 iterations used for smoothing. We use 20,000 draws for the variance estimator. The parameters are initialized at the unconditional choice alternative proportions and an equicorrelated covariance matrix. SML is implemented using an off-the-shelf optimization routine with the analytical score of the simulated likelihood derived in \citet{bolduc1999practical}. Following \citet{gates2006mata}, the number of simulation draws is set to $(J+1)\times 50$. MSS is implemented using an off-the-shelf root finding algorithm where the score vector equation is evaluated using the Gibbs resampling simulator proposed in \cite{hajivassiliou1998method}. We vary the number of simulation draws \(R\in\{10,100\}\) and the number of Gibbs resampling steps \(G\in\{1,5,10,15,20\}\). When $J=2$, we compute the exact MLE by numerical integration as described in Supplemental Appendix~\ref{A:mnp_MLE}.

\subsection{Simulation noise and comparison with MSS}\label{sec:MSS_Emp}
In this section, we compare the simulation noise in SEGA against that in MSS. As discussed in Section \ref{sec:MSS_discuss}, MSS can deliver unstable parameter estimates if the number of simulations $R$ is not chosen large enough, or if the Gibbs sampler used to estimate the scores does not draw from the stationary distribution. 

To isolate the simulation noise in the SEGA and MSS algorithms, we consider one fixed simulated dataset with $n = 10{,}000$ observations and repeatedly estimate the model using SEGA and MSS. Since the data are held fixed, variation across repetitions reflects algorithmic noise rather than sampling variation. We compare SEGA and MSS to the MLE estimate with $J=2$, and the estimate obtained from SML with $J=3$, which provides an estimate close to the MLE in this setting. 

Figure~\ref{fig:HMcompare_densities} shows the price coefficient estimates across 100 repetitions. The results are qualitatively similar across all parameters. The horizontal line denotes the MLE or SML estimate. The gold boxplots correspond to SEGA, while the grey boxplots correspond to MSS. Panels~(a) and~(b) report results for MSS with \(R=10\) simulation draws in the three-choice and four-choice models, respectively. Panels~(c) and~(d) repeat the comparison for MSS with \(R=100\) simulation draws.

\begin{figure}[tb!]
\caption{Simulation noise in SEGA and MSS estimates}
\centering
\includegraphics*[width=\textwidth,trim=1.5cm 0.8cm 1.5cm 0.6cm,clip]{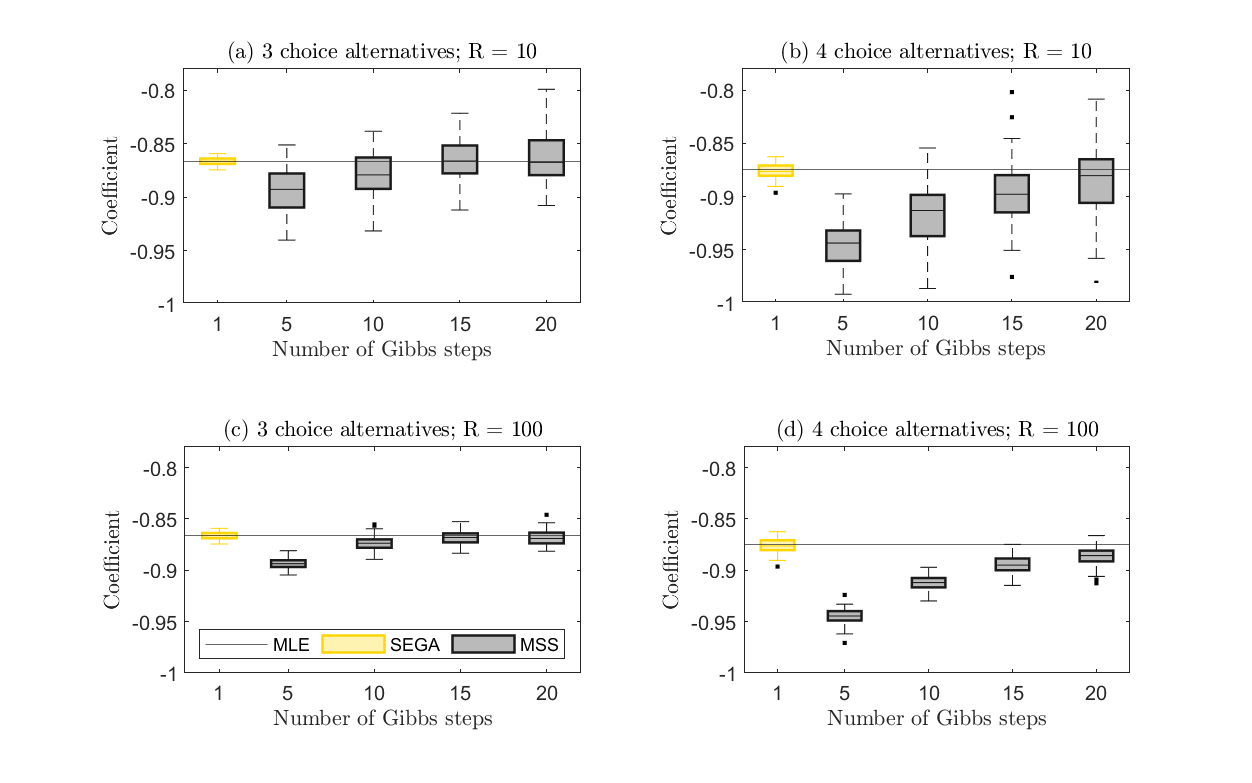}
\begin{flushleft}
\footnotesize
Gold and grey boxplots correspond to parameter estimate variation in SEGA and MSS, respectively. The boxplots are constructed from 100 repetitions. The horizontal lines denote the exact MLE estimate in (a) and (c) and the SML estimate in (b) and (d). Panels (a) and (b) show MSS results using $R=10$ and (c) and (d) using $R=100$ simulations. Panels (a) and (c) report estimates from the three-choice, and (b) and (d) from the four-choice MNP model.  
\end{flushleft}
\label{fig:HMcompare_densities}
\end{figure}

The figure shows that MSS is highly sensitive to its tuning parameters. When \(R=10\), the MSS estimates are substantially dispersed across repetitions, indicating substantial Monte Carlo variability in the estimator of the score function used by MSS. Increasing the number of Gibbs-resampling steps reduces bias, but does not remove the dispersion. This pattern is especially visible in the four-choice model in Panel~(b), where insufficient resampling produces estimates that are far from the SML estimate. Increasing the number of simulation draws to \(R=100\) reduces the variability of MSS, as shown in Panels~(c) and~(d). However, the remaining bias still depends strongly on the number of Gibbs-resampling steps, especially in the higher-dimensional model. 

In contrast, the SEGA estimates are tightly concentrated around the MLE and SML estimates. This stability arises from the sequential smoothing in SGA, which averages simulation noise progressively over iterations to reduce their variability. In contrast, MSS attempts to reduce simulation noise within each iteration, a strategy that substantially limits applicability in large-scale settings.

The experiment illustrates that the instability of MSS noted in the literature can be directly attributed to its sensitivity of its tuning parameters, together with its dependence on the initialization of the latent variables and model parameters. The design of this experiment is deliberately favorable to MSS: the method is initialized at the true parameter value, whereas SEGA is initialized at default values. When initialized at default values, MSS frequently encounters numerical failures in the root-finding algorithm; for example, with \(G=10\) and \(R=100\), approximately 30\% of runs fail. Given the implementation challenges associated with MSS, we do not include it in the following experiments.

%\subsection{Sampling variability}
\subsection{Sampling behavior and comparison with SML}
The second set of experiments studies sampling variability. We generate $1{,}000$ data sets and compare the sampling distributions of SEGA and SML, using the exact MLE as a benchmark when it is computationally available. Figure~\ref{fig:MCdens28} compares the repeated-sampling distributions of SEGA and SML for the price coefficient in the MNP model. Panels~(a) and~(c) report the low-dimensional case with \(J=2\), for which the exact MLE can also be computed by numerical integration. Panels~(b) and~(d) report the higher-dimensional case with \(J=6\), for which exact likelihood evaluation is computationally infeasible. The top row uses \(n=1{,}000\), while the bottom row uses \(n=20{,}000\). The vertical line indicates the true parameter value.

\begin{figure}[tb!]
\caption{Monte Carlo distributions coefficient estimates}
\centering
\includegraphics*[width=\textwidth]{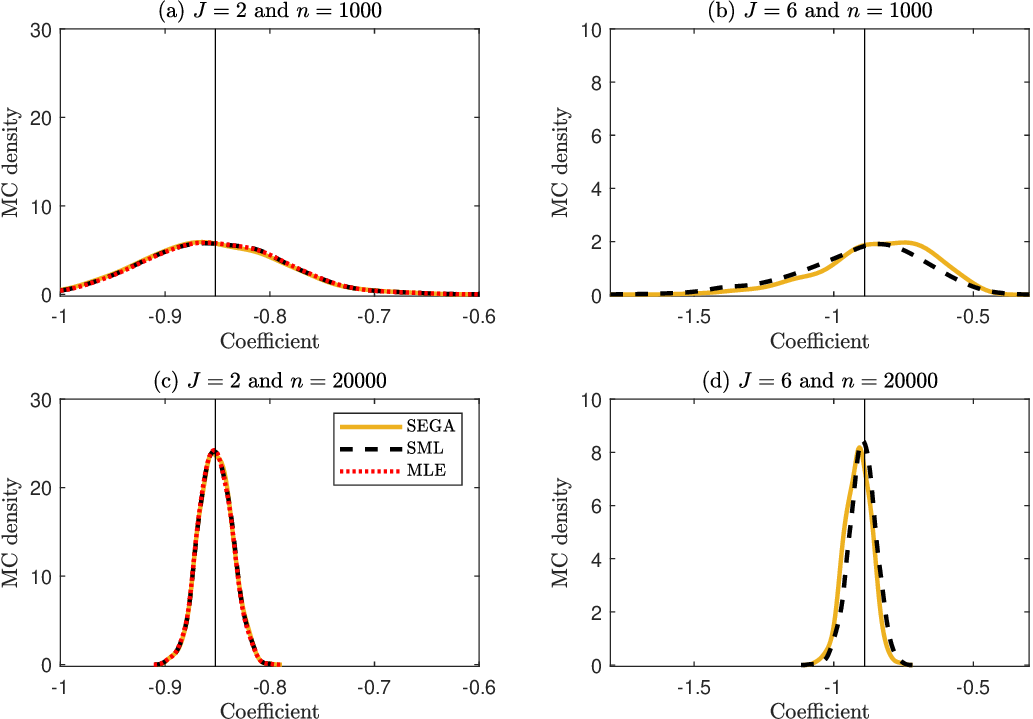}
\begin{flushleft}
\footnotesize
This figure shows the Monte Carlo distributions of the estimates for $\beta_J$ by SEGA (solid line) and SML (dotted line). The panels correspond to Monte Carlo experiments with $J=2,6$ and $n=1000,20000$. The vertical lines indicate the parameter value in the data generating process.
\end{flushleft}
\label{fig:MCdens28}
\end{figure} 

In the low-dimensional case, the SEGA, SML, and exact MLE distributions are virtually indistinguishable. This is visible in Panel~(a) and (c), where all three distributions are concentrated around the true value. These results show that, when the exact MLE is available as a benchmark, SEGA reproduces its sampling behavior closely. The higher-dimensional case is more challenging. In Panel~(b), where \(J=6\) and \(n=1{,}000\), the SEGA and SML distributions are both more dispersed and display skewness. The two distributions are not identical in this small-sample, higher-dimensional setting, reflecting the greater difficulty of estimating the MNP model when the latent utility dimension increases. However, Panel~(d) shows that this difference largely disappears when the sample size increases to \(n=20{,}000\). 

The patterns in Figure~\ref{fig:MCdens28} are qualitatively consistent with our theoretical results. In Panels~(a) and~(c), where the exact MLE is available, the sampling distributions of SEGA and MLE are virtually indistinguishable, as suggested by the asymptotic equivalence result in Theorem~\ref{thm:limit}. Notably, the finite-sample concentration result in Theorem~\ref{corr:smoothed} only controls the first two moments of the distance between SEGA and the infeasible MLE. That is, with reference to panels (b)-(d), there is no reason to suspect that the sampling distributions of SEGA and MLE will exactly agree with small-to-moderate values of $n$ used in this experiment, beyond the first two moments. To further complicate matters, the discrepancies may also be explained by the fact that the sampling distribution of the SML estimator is itself an approximation to that of the exact MLE.

Table~\ref{tab:MCstat28} summarizes the same comparison in terms of absolute bias and root mean squared error (RMSE) with respect to the parameter values in the data generating process, and computation time. The bias and RMSE results confirm the message from Figure~\ref{fig:MCdens28}. In the $J=2$ designs, SEGA, SML, and the exact MLE have nearly identical accuracy. In the $J=6$ designs, SEGA and SML also deliver similar RMSEs, with only small differences in absolute bias. Thus, the stochastic score approximation used by SEGA does not lead to a meaningful loss of accuracy, as measured by bias and RMSE, as predicted by our theoretical results.

\begin{table}[tb!]
  \centering
  \caption{Monte Carlo bias, RMSE, and computation time}
  \begin{threeparttable}
        \begin{tabular}{llrrrrrr}
        \toprule \toprule
          &       & \multicolumn{3}{c}{$J=2$} & \multicolumn{3}{c}{$J=6$} \\
            \cmidrule(lr){3-5} \cmidrule(lr){6-8}
    \multicolumn{1}{l}{$n$} &       & Bias & RMSE  & Time    & Bias & RMSE  & \multicolumn{1}{l}{Time} \\
    \midrule
1,000  & SEGA  &  0.021  &  0.162  &  0.085  &  0.044  &  0.428  &  0.424  \\
          & SML   &  0.007  &  0.145  &  0.028  &  0.116  &  0.540  &  2.175  \\
          & MLE   &  0.010  &  0.146  &  0.275  &      -  &      -  &      -  \\
          \cmidrule(lr){2-8}
   20,000 & SEGA  &  0.001  &  0.030  &  0.552  &  0.020  &  0.097  &  3.058  \\
          & SML   &  0.002  &  0.029  &  0.261  &  0.008  &  0.088  & 22.821  \\
          & MLE   &  0.002  &  0.029  &  7.485  &      -  &      -  &      -  \\
          \bottomrule \bottomrule
    \end{tabular}%
\begin{tablenotes}
\footnotesize
\item This table shows the absolute bias (Bias) and root mean squared error (RMSE) with respect to the parameter values in the data generating process, averaged across all parameters in the multinomial probit model. Furthermore, it reports the average computation time in minutes for each method across the replications in each experiment. 
\end{tablenotes}
\end{threeparttable}
  \label{tab:MCstat28}
\end{table}

The main difference between the methods is computational. When \(J=2\), SML is faster than SEGA because the simulated likelihood is relatively cheap to evaluate. However, the computational cost of SML grows rapidly with the number of alternatives. For \(J=6\) and \(n=1{,}000\), SML is already more than five times slower than SEGA. For \(J=6\) and
\(n=20{,}000\), SML is more than seven times slower. This difference reflects the distinct way in which the methods handle integration over the latent utilities. SML approximates choice probabilities at each likelihood evaluation, and this becomes increasingly costly as the latent utility dimension increases. SEGA instead uses a single draw of the latent variables at each iteration and averages the simulation noise over the stochastic-gradient path.

\subsection{Inferential accuracy}
The final experiment evaluates whether the standard errors implied by the asymptotic distribution in Theorem~\ref{thm:limit} deliver well-calibrated confidence intervals in finite samples. We compute coverage rates across 1,000 Monte Carlo replications for experiments with $J=2$ and a correctly specified model, $J=2$ and a misspecified model, and a correctly specified but higher-dimensional setting with $J=6$. 

The misspecified $J=2$ design generates from the same MNP model but the covariance matrix is made heteroskedastic: $\Sigma_i = \left(0.1+0.5{z_{i1}^2}\right)^2 \left(0.5 I_J + 0.5 \iota_J\iota_J^\top\right)$. The estimated model nevertheless imposes the homoskedastic MNP covariance specification in \eqref{eq:Z_i}. Hence, the target parameter is the pseudo-true value  $\theta_0=\arg\max_{\theta\in\Theta}E\{\ell_i(\theta)\}$, rather than the parameter value used to generate the data. We approximate this pseudo-true parameter by computing the exact MLE on 100 independent samples of size 100,000 and averaging the resulting estimates. 

Table~\ref{tab:coverage} reports the coverage rates for the experiments with $J=2$. For the correctly specified $J=2$ model, coverage is close to the nominal 95\% level for all parameters. At $n=1{,}000$, the coverage rates range from 0.91 to 0.95, and they improve further at $n=20{,}000$, where all rates lie between 0.94 and 0.96. 
The misspecified $J=2$ design delivers similarly accurate coverage: the rates are close to 0.95 at both sample sizes, indicating that sandwich-based SEGA inference remains well calibrated when the target is interpreted as the pseudo-true parameter.

\begin{table}[tb!]
\centering
\caption{Empirical coverage of 95\% SEGA confidence intervals}
\label{tab:coverage}
\begin{threeparttable}
\begin{tabular}{rccccc}
\toprule\toprule
Sample size & \(\beta_1\) & \(\beta_2\) & \(\beta_p\) & \(\Sigma_{12}\) & \(\Sigma_{22}\) \\
\hline
&\multicolumn{5}{c}{Correct model specification} \\
\cline{2-6}
\(1{,}000\)  & 0.94 & 0.93 & 0.95 & 0.91 & 0.91 \\
\(20{,}000\) & 0.95 & 0.96 & 0.94 & 0.95 & 0.95 \\
\hline
&\multicolumn{5}{c}{Model misspecification} \\
\cline{2-6}
\(1{,}000\)  & 0.96 & 0.94 & 0.96 & 0.94 & 0.93 \\
\(20{,}000\) & 0.94 & 0.95 & 0.94 & 0.95 & 0.94 \\
 \bottomrule \bottomrule
\end{tabular}
\begin{tablenotes}
\footnotesize
\item This table reports empirical coverage rates of nominal 95\% confidence intervals for the SEGA estimator across 1,000 Monte Carlo replications. The first panel reports results under correct specification and the second panel under misspecification. 
\end{tablenotes}
\end{threeparttable}
\end{table}

In the high-dimensional setting, the Monte Carlo experiment indicates that larger sample sizes are required for accurate empirical coverage. For the $J=6$ design, the model contains 27 parameters, and the individual coverages range between 0.83 and 0.99 with $n=1{,}000$, and 0.89 and 0.94 with $n=20{,}000$, while the median coverage is 0.93 and 0.92, respectively. As the number of alternatives increases, the likelihood involves higher-dimensional integrals, and accurate estimation of the sampling variability requires substantially more data. Our empirical application with the multinomial probit model contains more than one million observations, so the large-sample regime is the relevant one for that setting. However, for this highly-nonlinear model, Monte Carlo validation at that scale is computationally infeasible.

\section{Conclusion}\label{sec:conclusion}
Latent variable models are central to empirical work in economics, but likelihood-based inference in these models is often limited by the need to integrate over high-dimensional latent variables.
This paper proposes SEGA, a stochastic estimation framework that uses Fisher’s identity to construct an unbiased estimate of the likelihood score from a single draw of the latent variables. By embedding this score in a stochastic gradient ascent algorithm, SEGA averages simulation noise over the optimization path rather than requiring accurate likelihood or score approximations at each parameter value.

The theoretical results show that SEGA has the same first-order inferential properties as the infeasible MLE. In finite samples, the SEGA estimator concentrates around the infeasible MLE as the number of iterations increases. Asymptotically, when the number of iterations grows appropriately with the sample size, the SEGA estimator is asymptotically equivalent to the infeasible MLE and inherits its first-order limiting distribution, including the usual sandwich form under misspecification. The numerical experiments support these theoretical results.

The empirical applications illustrate the value of SEGA for applied researchers. In the multinomial probit application, SEGA allows us to estimate an unrestricted model of brand choice using more than one million purchase observations, construct confidence intervals for own- and cross-price elasticities, and test covariance restrictions that are commonly imposed for tractability. In the random effects Tobit application, SEGA allows us to estimate a household-level demand model with more than 200,000 households observed over 62 weeks, and to conduct inference on heterogeneity in baseline demand, price sensitivity, and state dependence. 

These applications show that SEGA makes likelihood-based inference feasible in LDVMs at scales where existing methods are impractical. By removing this computational barrier, SEGA expands the set of empirical questions that applied researchers can study with flexible limited dependent variable models.
%\newpage 
\appendix
\section{Proofs of Main Results}\label{A:proofs}
The proof of Theorem~\ref{corr:smoothed} makes use of the following result, which derives an expected per-iteration bound for each SEGA iterate $\theta^{(k)}_n$. This result delivers a sharper per-iteration rate of control for SEGA iterates than what is known to be achievable for SGA algorithms. The reasons for this tighter bound are discussed in the main text. 
\begin{lemma}\label{thm:main}
Under Assumptions \ref{ass:mds}-\ref{ass:iters},
$$
\E\|\theta^{(k)}_n-\hat\theta_n\|^p\lesssim \left\{\exp(-  k^{1-\alpha})\E\|\theta^{(0)}_n-\hat\theta_n\|^2_2\right\}^{p/2} +\left({nk^{\alpha}}\right)^{-p/2},\quad p\in\{1,2\}.$$
\end{lemma}
%\begin{cases}
%\exp\{-C_3 k^{1-\alpha}\frac{p}{2}\} \sqrt{\E\|\delta_0\|^2_2}+\left(\frac{{C_3}}{nk^{2\alpha}}\right)^{p/2}&\quad p=1\\ \exp\{-C_3 k^{1-\alpha}/2\} \sqrt{\E\|\delta_0\|^2_2}+\frac{\sqrt{C_3}}{\sqrt{n}k^{\alpha}}&\quad p=2
%\end{cases}

\begin{proof}[Proof of Lemma \ref{thm:main}]
To simplify notation, let $\zeta_k=\widetilde{z}_{1:n}^{(k)}$, where we recall that $\widetilde{z}_{1:n}^{(k)}\sim p_{\theta^{(k)}_n}(\cdot\mid d_{1:n})$. Define $\delta_{k}:=\theta^{(k)}_n-\hat\theta_n$, and define the shifted gradients $\widetilde{m}^{(k)}_n(\delta_{k})=\widehat{m}_n(\delta_{k}+\hat\theta;\zeta_k)$ and $\widetilde{m}^{(k)}_n(0)=\widehat{m}_n(\hat\theta;\zeta_k)$. 
We rewrite the recursion in \eqref{eq:new} as
$$
\delta_{k}=\delta_{k-1}+\eta_n^{(k)}\widetilde{m}^{(k-1)}_n(\delta_{k-1}).
$$
Now, we consider the behavior of $\|\delta_{k}\|_2^2$, which can be written as
\begin{flalign}
\|\delta_{k}\|_2^2&=\|\delta_{k-1}\|_2^2+\{\eta_n^{(k)}\}^2\|\widetilde{m}^{(k-1)}_n(\delta_{k-1})\|^2_2+2\eta_n^{(k)}\langle\widetilde{m}^{(k-1)}_n(\delta_{k-1}),\delta_{k-1}\rangle\nonumber\\&=\text{A}+\text{B}+\text{C}.\label{eq:square}
\end{flalign}
We now consider terms B and C separately. 
\\

\noindent \textbf{Term B}. 
%First, consider that 
%\begin{flalign*}
%&\|\widetilde{m}^{(k-1)}_n(\delta_{k-1})\|^2_2\\&=\|\widetilde{m}^{(k-1)}_n(\delta_{k-1})-\widetilde{m}^{(k-1)}_n(0)+\widetilde{m}^{(k-1)}_n(0)\|^2_2\\&=\|\widetilde{m}^{(k-1)}_n(0)\|^2_2+\|\widetilde{m}^{(k-1)}_n(\delta_{k-1})-\widetilde{m}^{(k-1)}_n(0)\|_2^2 +2\langle\widetilde{m}^{(k-1)}_n(\delta_{k-1})-\widetilde{m}^{(k-1)}_n(0),\widetilde{m}^{(k-1)}_n(0) \rangle
%\\&\le 
%\|\widetilde{m}^{(k-1)}_n(0)\|^2_2+\|\widetilde{m}^{(k-1)}_n(\delta_{k-1})-\widetilde{m}^{(k-1)}_n(0)\|_2^2 +2\|\widetilde{m}^{(k-1)}_n(\delta_{k-1})-\widetilde{m}^{(k-1)}_n(0)\|_2\|\widetilde{m}^{(k-1)}_n(0)\|_2
%\\&\le\|\widetilde{m}^{(k-1)}_n(0)\|^2_2+\|\widetilde{m}^{(k-1)}_n(\delta_{k-1})-\widetilde{m}^{(k-1)}_n(0)\|_2^2 +[\|\widetilde{m}^{(k-1)}_n(\delta_{k-1})-\widetilde{m}^{(k-1)}_n(0)\|^2_2+\|\widetilde{m}^{(k-1)}_n(0)\|^2_2]\\&= 2\|\widetilde{m}^{(k-1)}_n(\delta_{k-1})-\widetilde{m}^{(k-1)}_n(0)\|^2_2+2\|\widetilde{m}^{(k-1)}_n(0)\|^2_2,
%\end{flalign*}where the first inequality comes from Cauchy-Schwarz and the second by applying Young's inequality to the last term. 

With similar notations, define $\widetilde{m}_n(\delta_{k-1})=m_n(\theta_n^{(k-1)})$, and $\widetilde{m}_n(0)=m_n(\hat\theta_n)$. Then,
\begin{flalign*}
&\|\widetilde{m}^{(k-1)}_n(\delta_{k-1})\|^2_2\\&=\|\widetilde{m}^{(k-1)}_n(\delta_{k-1})-\widetilde{m}_n(\delta_{k-1})+\widetilde{m}_n(\delta_{k-1})\|^2_2\\&= \|\widetilde{m}^{(k-1)}_n(\delta_{k-1})-\widetilde{m}_n(\delta_{k-1})\|^2_2+\|\widetilde{m}_n(\delta_{k-1})\|^2_2+2\langle\widetilde{m}^{(k-1)}_n(\delta_{k-1})-\widetilde{m}_n(\delta_{k-1}),\widetilde{m}_n(\delta_{k-1})\rangle \\&=\|\widetilde{m}^{(k-1)}_n(\delta_{k-1})-\widetilde{m}^{}_n(\delta_{k-1})\|^2_2+\|\widetilde{m}^{}_n(\delta_{k-1})-\widetilde{m}^{}_n(0)\|^2_2+2\langle\widetilde{m}^{(k-1)}_n(\delta_{k-1})-\widetilde{m}_n(\delta_{k-1}),\widetilde{m}_n(\delta_{k-1})\rangle\\&=\mathrm{B.1+B.2+B.3}. %\\&\le \|\widetilde{m}^{(k-1)}_n(\delta_{k-1})-\widetilde{m}_n(\delta_{k-1})\|^2_2+L^2\|\delta_{k-1}\|^2+2\langle\widetilde{m}^{(k-1)}_n(\delta_{k-1})-m_n(\delta_{k-1}),m_n(\delta_{k-1})\rangle
\end{flalign*}Under Assumption \ref{ass:mds}(i), $$
\E[\mathrm{B.3}\mid\mathcal{F}_{k}]=\E[\langle\widetilde{m}^{(k)}_n(\delta_{k})-\widetilde{m}_n(\delta_{k}),\widetilde{m}_n(\delta_{k})\rangle\mid\mathcal{F}_{k}]=0.$$From Assumption \ref{ass:mds}(ii) and the definition of $\widetilde{m}_n(\delta_k)$, we have
\begin{flalign*}
&\E[\mathrm{B.1}\mid\mathcal{F}_{k}]\\&=%\E [\|\widetilde{m}^{}_n(\delta_{k})-\widetilde{m}^{}_n(\delta_{k})\|^2_2\mid\mathcal{F}_{k}]\\&
\E\left\{ \left\|\frac{1}{n}\sum_{i=1}^{n}g_i(\theta_n^{(k)};z_i^{(k)})-\frac{1}{n}\sum_{i=1}^{n}g_i(\theta_n^{(k)})\right\|^2_2\mid\mathcal{F}_{k}\right\}\\&= \E \left\{\left\|\frac{1}{n}\sum_{i=1}^{n}\left\{g_i(\theta_n^{(k)};z_i^{(k)})-\E[g_i(\theta_n^{(k)};z_i^{(k)})\mid\mathcal{F}_{k}]\right\}\right\|^2_2\mid\mathcal{F}_{k}\right\}\\&\le \frac{1}{n^2}\sum_{i=1}^{n}\E \left\{\|g_i(\theta^{(k)}_n;z_i^{(k)})-\E[g_i(\theta^{(k)}_n;z_i^{(k)})\mid\mathcal{F}_{k}]\|^2\mid\mathcal{F}_{k}\right\}\\&+\frac{1}{n^2}\sum_{i\ne j=1}^{n}\E\left[ \{g_i(\theta^{(k)};z_i^k)-\E[g_i(\theta^{(k)};z_i^k)\mid\mathcal{F}_k]\}^\top \{g_j(\theta^{(k)};z_j^k)-\E[g_j(\theta^{(k)};z_j^k)\mid\mathcal{F}_k]\}\mid\mathcal{F}_k\right]\\&\le \nu^2/n +0,
\end{flalign*}where the zero term comes from conditional independence of $z_i^{(k)}$ and $z_j^{(k)}$ conditional on $\mathcal{F}_{k}$. 

Under Assumption \ref{ass:conv}(ii), we have that 
\begin{equation*}
\|\widetilde{m}_n(\delta_{k-1})-\widetilde{m}_n(0)\|^2_2\le L^2\|\delta_{k-1}\|^2_2.%\label{eq:lipz_term}
\end{equation*}
%Further, under Assumption \ref{ass:mds}(ii), 
%\begin{flalign*}
%\E_{}[\|\widetilde{m}^{(k-1)}_n(0)\|^2_2\mid\mathcal{F}_{k-1}]&=\E_{}[\|\frac{1}{n}\sum_{i=1}^{n}\widetilde{g}(\hat\theta_n,z_i)-m_n(\hat\theta_n)\|^2_2\mid\mathcal{F}_{k-1}]
%\nonumber\\&\le \frac{1}{n^2}\sum_{i=1}^{n}\E_{}[\|\widetilde{g}(\hat\theta_n,z_i)-m_n(\hat\theta_n)\|^2_2\mid\mathcal{F}_{k-1}]
%\nonumber\\&\le \nu^2/n\label{eq:var_term}.
%\end{flalign*}

Hence, taking conditional expectations of $\|\widetilde{m}^{(k-1)}_n(\delta_{k-1})\|^2_2$ and applying the above display equations yields 
\begin{flalign}\label{eq:termB}
\E\left[\|\widetilde{m}^{(k-1)}_n(\delta_{k-1})\|^2_2\mid \mathcal{F}_{k-1}\right]%\le&     2\E\left[\|\widetilde{m}^{(k-1)}_n(\delta_{k-1})-\widetilde{m}^{(k-1)}_n(0)\|^2_2\mid \mathcal{F}_{k-1}\right]+\\&2\E\left[\|\widetilde{m}^{(k-1)}_n(0)\|^2_2\mid \mathcal{F}_{k-1}\right]\nonumber
\le L^2\|\delta_{k-1}\|^2_2+2\nu^2/n.
\end{flalign}

\noindent \textbf{Term C}. 
First, we write 
\begin{flalign*}
2\langle\widetilde{m}^{(k-1)}_n(\delta_{k-1}),\delta_{k-1} \rangle=2\langle\widetilde{m}^{(k-1)}_n(\delta_{k-1})-\widetilde{m}_n(\delta_{k-1}),\delta_{k-1} \rangle+2\langle \widetilde{m}_n(\delta_{k-1}),\delta_{k-1} \rangle ,  
\end{flalign*}where, with similar notions to $\widetilde{m}_n^{(k)}(\delta_{k-1})$, we set $\widetilde{m}_n(\delta_{k-1})=m_n(\delta_{k-1}+\hat\theta)$. 

Taking expectations, and using Assumption \ref{ass:mds}(i), yields 
\begin{flalign*}
&2\E_{{}}\left[\langle\widetilde{m}^{(k-1)}_n(\delta_{k-1}),\delta_{k-1} \rangle\mid \mathcal{F}_{k-1}\right]\\&=2\langle\E\left[\widetilde{m}^{(k-1)}_n(\delta_{k-1})-\widetilde{m}_n(\delta_{k-1})\mid \mathcal{F}_{k-1}\right],\delta_{k-1} \rangle+2\langle \widetilde{m}_n(\delta_{k-1}),\delta_{k-1} \rangle  
\\&=2\langle \widetilde{m}_n(\delta_{k-1}),\delta_{k-1} \rangle.
\end{flalign*}
Applying the strong concavity of $\ell_n(\theta)$, Assumption \ref{ass:conv}, and Lemma \ref{lem:concavity} in Supplemental Appendix~\ref{A:lemmas} to obtain, for $\mu>0$,
\begin{flalign*}
  \langle \widetilde{m}_n(\delta_{k-1}),\delta_{k-1} \rangle    \le-\frac{\mu}{2}\|\delta_{k-1}\|^2_2 .
\end{flalign*}
Hence, 
\begin{flalign}
2\E\left[\langle\widetilde{m}^{(k-1)}_n(\delta_{k-1}),\delta_{k-1} \rangle\mid \mathcal{F}_{k-1}\right]\le -\mu\|\delta_{k-1}\|^2_2 . \label{eq:termC}
\end{flalign}

\noindent\textbf{A+B+C.} Taking expectations of equation \eqref{eq:square} and applying equations \eqref{eq:termB} and \eqref{eq:termC} then yields 
\begin{flalign*}
\E[\|\delta_{k}\|^2_2\mid\mathcal{F}_{k-1}]&\le \|\delta_{k-1}\|^2_2 +L^2[\eta_{n}^{(k)}]^2\|\delta_{k-1}\|^2_2+2[\eta_{n}^{(k)}]^2\nu^2/n-\mu[\eta_{n}^{(k)}]\|\delta_{k-1}\|_2^2\\&=\|\delta_{k-1}\|_2^2\left\{1-\mu[\eta_{n}^{(k)}]+C_1[\eta_{n}^{(k)}]^2\right\}+2[\eta_{n}^{(k)}]^2\nu^2/n,    
\end{flalign*}where $C_1=L^2$.
%\dtf{In the second case, the constant 1 just becomes 3 and nothing else changes.  }
%There exists some constant $C_1$ such that 
%\begin{flalign*}
%\E[\|\delta_{k}\|^2_2\mid\mathcal{F}_{k-1}]&\le\|\delta_{k-1}\|_2^2\left\{1-\mu[\eta_{n}^{(k)}]+2C_1[\eta_{n}^{(k)}]^2\right\}+2[\eta_{n}^{(k)}]^2\nu^2.
%\end{flalign*}
Use the form of the learning rate $\eta_n^{(k)}$, and the fact that $\eta_n$ is $\mathcal{F}_0$ measurable to write:
$$
\E[\|\delta_{k}\|^2_2\mid\mathcal{F}_{k-1}]\le\|\delta_{k-1}\|_2^2\left\{1-\frac{\mu\eta_{n}}{k^\alpha}+2C_1\frac{\eta_n^2}{k^{2\alpha}}\right\}+2\frac{\nu^2\eta_n^2}{k^{2\alpha}n}.
$$
Now, for some $C>0$, let
$$
k^\ast=\min\left\{k\in\mathbb{N}:C\frac{\eta_n^2}{k^{2\alpha}}\le \frac{\mu\eta_n}{2k^\alpha}\right\},
%k^\ast=\min\left\{k\in\mathbb{N}:C_1\frac{\eta_n^2}{k^{2\alpha}}\le \frac{\mu\eta_n}{2k^\alpha},\;\;\frac{\mu\eta_nk^{(1-\alpha)}}{4}\ge 2\alpha\log(k)\right\},
$$
which exists for $k$ large enough so long as $0<\eta_n<\infty$. For any $k\ge k^\ast$, we have 
\begin{flalign*}
\E[\|\delta_{k}\|^2_2\mid\mathcal{F}_{k-1}]&\le\|\delta_{k-1}\|_2^2\left\{1-\frac{1}{2}\mu[\eta_{n}^{(k)}]\right\}+2[\eta_{n}^{(k)}]^2\nu^2/n.
\end{flalign*}
Applying Lemma B.2 in \cite{chen2020statistical} (Lemma \ref{lem:chen2020} in Supplemental Appendix~\ref{A:lemmas}), we have, for any $m\le k-1$,
\begin{flalign*}
\E[\|\delta_{k}\|^2_2\mid\mathcal{F}_{m}]%&\le\exp\left\{-\frac{\mu}{4}\sum_{i=m+1}^{k}\eta^{(i)}_n)\right\}+(4\nu^2/\mu)[\eta_n^{(m)}].
&\le \|\delta_{m}\|_2^2\exp\left\{-\frac{\mu}{2}(k-m)\eta_n^{(k)}\right\}+(4\nu^2/\mu)[\eta_n^{(m)}]/n.%\exp\left\{-\frac{\mu}{4}\sum_{i=m}^{k}\eta^{(i)}_n\right\}+C_2\eta_n^{1+\gamma}m^{-(1+\gamma)}.
\end{flalign*}

Now, set $m=\lfloor k/2 \rfloor>k^\ast$, and note that 
$
\frac{\mu}{2}(k-m)\frac{\eta_n}{k^\alpha}\ge \frac{\mu\eta_n}{4}k^{1-\alpha}
$ from which we have 
$$
\exp\left\{-\frac{\mu}{2}(k-m)\frac{\eta_n}{k^\alpha}\right\}\le \exp\left\{-\frac{\mu\eta_n}{4}k^{1-\alpha}\right\}.
$$Set $C_1:=\frac{\mu\eta_n}{4}$, and $C_2:=4\nu^2/\mu$. When $m=\lfloor k/2 \rfloor>k^\ast$, use the definition of $\eta^{(k)}_n=\eta_n/k^\alpha$ to obtain 
\begin{flalign*}
\E\left[\|\delta_{k} \|^2_2\mid \mathcal{F}_{m}\right]&\le \|\delta_{k/2}\|^2_2\exp\left\{-C_1k^{1-\alpha}\right\}+C_2\frac{\eta_n}{k^{\alpha}n}.%\le k^{-\alpha}(\|\delta_{k/2}\|^2_2+C_2).
\end{flalign*}

Applying the above, and the law of iterated expectations, for some  $C_3>0$ we obtain 
$$
\E_{}\left[\|\delta_{k} \|^2_2\right]\le \exp\{-C_3 k^{1-\alpha}\} \E\|\delta_0\|^2_2+\frac{C_3}{nk^{\alpha}}.
$$
To obtain the result for $p=1$, use Cauchy-Schwarz to see that 
\begin{flalign*}
\E_{}\left[\|\delta_{k} \|_2\right]&\le \sqrt{\E\left[\|\delta_{k} \|^2_2\right]}\\&\le \sqrt{\exp\{-C_3 k^{1-\alpha}\} \E\|\delta_0\|^2_2+\frac{C_3}{nk^{\alpha}}}\\&\le \exp\{-C_3 k^{1-\alpha}/2\} \sqrt{\E\|\delta_0\|^2_2}+\frac{\sqrt{C_3}}{\sqrt{n}k^{\alpha/2}}.
\end{flalign*}

\end{proof}

Lemma \ref{thm:main} allows us to give a simple and direct proof of Theorem \ref{corr:smoothed}. The proof bounds the error of the average estimator $\bar\theta_n$ by summing the bounds obtained from the per-iteration estimators. As such, the resulting bound is loose and we suspect that it can be sharpened if an alternative proof technique is used. 

\begin{proof}[Proof of Theorem~\ref{corr:smoothed}]
Now consider $\overline\theta_n$ and set $p=1$ for simplicity. The result for $p=2$ follows the same way. For each term $\|\theta^{(j)}_n-\hat\theta_n\|$ we have, for $j>1$, %by Theorem \ref{thm:main}, 
$$
\E\left\|\theta^{(j)}_n-\hat\theta_n\right\|\le \exp\{-C j^{1-\alpha}\} \sqrt{\E\|\delta_0\|^2_2}+\frac{C}{\sqrt{n}j^{\alpha/2}},
$$for some constant $C>0$, and where $\delta_0=\theta^{(0)}_n-\hat\theta_n$. Use the above and the fact that the norm $\|\cdot\|$ is convex to obtain
\begin{flalign*}
\E\left\|\frac{1}{k-k_0}\sum_{j=k_0+1}^{k}\theta^{(j)}_n-\hat\theta_n\right\|&\le \frac{1}{k-k_0}\sum_{j=k_0+1}^{k}\E\|\theta^{(j)}_n-\hat\theta_n\|\\&\le\sqrt{\E\|\delta_0\|^2_2}\frac{1}{k-k_0}\sum_{j=k_0+1}^{k}e^{-Cj^{1-\alpha}}+\frac{1}{\sqrt{n}(k-k_0)}\sum_{j=k_0+1}^{k}{j^{-\alpha/2}}.
\end{flalign*}
For the second term, use the approximation 
$
\sum_{j=k_0+1}^{k}{-j^{\alpha/2}}\leq C \frac{k^{1-\alpha/2}}{1-\alpha/2},
$ and the fact that $k-k_0\asymp k$ to deduce that 
\begin{flalign*}
\E\left\|\frac{1}{k-k_0}\sum_{j=k_0+1}^{k}\theta^{(j)}_n-\hat\theta_n\right\|\le\sqrt{\E\|\delta_0\|^2_2}\frac{1}{k-k_0}\sum_{j=k_0+1}^{k}e^{-Cj^{1-\alpha}}+\frac{C}{(1-\alpha/2)\sqrt{n}k^{\alpha/2}}.
\end{flalign*}For the first term, we have that $\sum_{j=k_0+1}^{k}e^{-Cj^{1-\alpha}}\le \frac{C^{-1 /(1-\alpha)}}{1-\alpha} \Gamma\left(\frac{1}{1-\alpha}, C k_0^{1-\alpha}\right)$, where $\Gamma(s,z)=\int_{z}^{\infty}t^{s-1}e^{-t}dt$ is the incomplete Gamma function. This bound is finite for fixed $k_0$ and $k_0\rightarrow\infty$. Hence, for some $C>0$ and $k-k_0\asymp k$, we have 
\begin{flalign*}
\E\left\|\frac{1}{k-k_0}\sum_{j=k_0+1}^{k}\theta^{(j)}_n-\hat\theta_n\right\|\le&\sqrt{\E\|\delta_0\|^2_2}\frac{\frac{C^{-1 /(1-\alpha)}}{1-\alpha} \Gamma\left(\frac{1}{1-\alpha}, C k_0^{1-\alpha}\right)}{k-k_0}+\frac{C}{(1-\alpha/2)\sqrt{n}k^{\alpha/2}}\\
&\lesssim k^{-1}\sqrt{\E\|\delta_0\|^2_2}+{n}^{-1/2}k^{-\alpha/2}.
\end{flalign*}

\end{proof}

\begin{proof}[Proof of Theorem \ref{thm:limit}]
Write 
$$
\sqrt{n}(\overline{\theta}_n-\theta_0)=\sqrt{n}(\hat\theta_n-\theta_0)+\sqrt{n}(\overline{\theta}_n-\hat\theta_n).
$$If we can show that $\sqrt{n}(\overline\theta_n-\hat\theta_n)=o_p(1)$, then the stated result follows by Slutsky's theorem. 

Fix $\delta>0$. From Theorem~\ref{corr:smoothed}, and Markov, 
\begin{flalign*}
\Pr\left[\sqrt{n}\|\overline\theta_n-\hat\theta_n\|>\delta\right]\le& \frac{\sqrt{n}}{\delta}\E\|\overline\theta_n-\hat\theta_n\|\\&\leq C \frac{1}{\delta}\left\{\frac{1}{k^{\alpha/2}}+{\left(\E\left[\|\theta^{(0)}_n-\hat\theta_n\|^2\right]\right)^{1/2}}\frac{\sqrt{n}}{k}\right\}.
\end{flalign*}
By Assumption \ref{ass:starting}, we have that 
$$
\E\|\theta^{(0)}_n-\hat\theta_n\|^2\le \E(C_n^2)=M<\infty.
$$
Fix $\epsilon>0$. Given $\delta>0$, since $k\gg\sqrt{n}$, there exists an $n'$ such that $$
\frac{C}{\delta}\left\{\frac{1}{k^{\alpha/2}}+M\frac{\sqrt{n'}}{k}\right\}<\epsilon,
$$ and so for $n\ge n'$, we have $\Pr\left[\sqrt{n}\|\overline\theta_n-\hat\theta_n\|>\delta\right]\le\epsilon$. Since $\epsilon$ and $\delta$ are arbitrary,  $\sqrt{n}\|\overline{\theta}_n-\hat\theta_n\|=o_p(1)$, and the result follows.
\end{proof}

\bibliographystyle{apalike} 
\bibliography{mle_mnp}

\newpage
\noindent
\setcounter{page}{1}
\begin{center}
	{\bf \Large{Supplemental Appendix for ``Scalable likelihood-based inference for limited dependent variable models''}}
\end{center}

\vspace{10pt}
\setcounter{equation}{0}
\setcounter{figure}{0}
\setcounter{table}{0}
\setcounter{section}{0}
\renewcommand{\thetable}{S\arabic{table}}
\renewcommand{\thefigure}{S\arabic{figure}}
\renewcommand{\thesection}{S\Alph{section}}
\renewcommand{\thesubsection}{S\Alph{section}.\arabic{subsection}}
\renewcommand{\theequation}{S\arabic{equation}}

\noindent
This appendix has three parts:

\begin{itemize}
	\item[] {\bf Part~SA}: Example 1: Multinomial probit model.
	\item[] {\bf Part~SB}: Example 2: Random effects Tobit model.
	\item[] {\bf Part~SC}: Supplementary lemma's.
\end{itemize}
\newpage

\section{Example 1: Multinomial probit model}

\subsection{SEGA implementation}\label{A:mnp_fisher}
First, we derive the exact expressions in the gradient of the log augmented likelihood with respect to $\theta=(\beta^\top,\xi^\top)^\top$ presented in \eqref{eq:grad_augpost_mnp}:
\begin{align}
    \nabla_\theta \log p_\theta(y_{1:n},z_{1:n}|x_{1:n}) =&
    \nabla_\theta \log p(y_{1:n}\mid z_{1:n})+\sum_{i=1}^n \nabla_\theta \log \phi_J(z_i;x_i\beta,\Sigma) \\
    =& \sum_{i=1}^n (\nabla_\beta \log \phi_J(z_i;x_i\beta,\Sigma)^\top,\nabla_\xi \log \phi_J(z_i;x_i\beta,\Sigma)^\top). \notag
\end{align}
Define ${\eta}_i = {z}_i-x_i{\beta}$. The gradient with respect to $\beta$ equals
\begin{align}
    \nabla_{\beta}\log \phi_{J}\left({z}_i;x_i{\beta},\Sigma\right)^\top& = {\eta}_i^\top\Sigma^{-1}x_i.
\end{align}
The gradient with respect to $\xi$ equals
\begin{align}\label{eq:Gexpres}
    \nabla_{\xi}\log \phi_{J}\left({z}_i;x_i{\beta},\Sigma\right)^\top &= \frac{\partial}{\partial  {\xi}}\left[-\frac{J}{2}\log(2\pi)+\frac{1}{2}\log (\text{det}(\Sigma^{-1}))-\frac{1}{2} {\eta}_i^\top\Sigma^{-1} {\eta}_i\right]\\
    & =  \frac{1}{2}\frac{\partial}{\partial  {\xi}}\log (\text{det}(\Sigma^{-1}))-\frac{1}{2}\frac{\partial}{\partial  {\xi}} {\eta}_i^\top\Sigma^{-1} {\eta}_i\nonumber\\
    & =  \frac{1}{2}\frac{1}{\text{det}(\Sigma^{-1})}\frac{\partial \text{det}(\Sigma^{-1})}{\partial \Sigma^{-1}}\frac{\partial \Sigma^{-1}}{\partial  {\xi}}-\frac{1}{2}\frac{\partial}{\partial  {\xi}} {\eta}_i^\top\Sigma^{-1} {\eta}_i\nonumber \\
    & =  \frac{1}{2}\frac{1}{\text{det}(\Sigma^{-1})}\frac{\partial \text{det}(\Sigma^{-1})}{\partial \Sigma^{-1}}\frac{\partial \Sigma^{-1}}{\partial  {\xi}}-\frac{1}{2}\left( {\eta}_i^\top\otimes {\eta}_i^\top\right)\frac{\partial}{\partial  {\xi}}\text{vec}(\Sigma^{-1})\nonumber\\
     & = \frac{1}{2}\frac{1}{\text{det}(\Sigma^{-1})}\frac{\partial \text{det}(\Sigma^{-1})}{\partial \Sigma^{-1}}\frac{\partial \Sigma^{-1}}{\partial  {\xi}}-\frac{1}{2}\left( {\eta}_i^\top\otimes {\eta}_i^\top\right)\frac{\partial \Sigma^{-1}}{\partial  {\xi}}\nonumber\\
     & =  \frac{1}{2}\frac{1}{\text{det}(\Sigma^{-1})}\frac{\partial \text{det}(\Sigma^{-1})}{\partial \Sigma^{-1}}\frac{\partial \Sigma^{-1}}{\partial C}\frac{\partial C}{\partial {\kappa}}\frac{\partial  {\kappa}}{\partial  {\xi}}-\frac{1}{2}\left( {\eta}_i^\top\otimes {\eta}_i^\top\right)\frac{\partial \Sigma^{-1}}{\partial C}\frac{\partial C}{\partial {\kappa}}\frac{\partial  {\kappa}}{\partial  {\xi}}\nonumber,
\end{align}
where $\frac{\partial \text{det}(\Sigma^{-1})}{\partial \Sigma^{-1}} = \text{det}(\Sigma^{-1})\text{vec}(\Sigma)^\top$ and $\frac{\partial \Sigma^{-1}}{\partial C} = \left(I_{J^2}+K_{J,J}\right)\left(C\otimes I_J\right)$,
with $K_{m,n}$ the commutation matrix of an $m\times n$ matrix. Since only the lower diagonal elements of $C$ depend on $ \kappa$, the derivative matrix $\frac{\partial C}{\partial  {\kappa}}$ is sparse. The non-sparse elements of this matrix can be constructed as $\frac{\partial \text{vech}(C)}{\partial  {\kappa}}$, where
$\left\{\frac{\partial \text{vech}(C)}{\partial  {\kappa}}\right\}_{l,j}=\frac{\partial \psi_{l}( {\kappa})}{\partial \kappa_{j}}$ with
\begin{align*}
\frac{\partial \psi_{l}( {\kappa})}{\partial \kappa_{j}}=
\left\{
\begin{array}{ll}
\sqrt{J}\cos\left(\kappa_{j}\right)\cos\left(\kappa_{l}\right)\prod_{s\in\{1,\dots,l-1\}\backslash j}\sin\left(\kappa_{s}\right) & \text{if } j<l \text{ and } l<d_C,\\
-\sqrt{J}\prod_{s=1}^{l}\sin\left(\kappa_{s}\right) & \text{if } j=l \text{ and } l<d_C,\\
\sqrt{J}\cos\left(\kappa_{j}\right)\prod_{s\in\{1,\dots,l-1\}\backslash j}\sin\left(\kappa_{s}\right) & \text{if } j<l \text{ and } l=d_C,\\
0 & \text{if otherwise.} 
\end{array}
\right.
\end{align*}
Second, we produce draws from $p_\theta(z_{1:n}\mid d_{1:n})$ in \eqref{eq:pz_mnp} using the Gibbs sampler proposed by \citet{geweke1991efficient}.
 
\subsection{Variance estimator implementation}\label{A:MNPse} 
First, we use the method proposed by \citet{botev2017normal} to generate i.i.d.\ draws from the truncated multivariate normals implied by $p_\theta(z_{1:n}\mid d_{1:n})$. 

Second, we derive the expressions for $\nabla_\theta \log p_\theta(   y_i,    z_i| x_i)$ and $\nabla_\theta^2 \log p_\theta(   y_i,    z_i| x_i)$ in \eqref{eq:Vhat} to compute standard errors for the estimates of $\beta$ and the distinct elements of $\Sigma$, with the normalization that the first diagonal element of $\Sigma$ is equal to one:
\begin{align*}
    \nabla_\theta \log p_\theta(   y_i,    z_i| x_i)=&\left(\nabla_\beta \log p_\theta(   y_i,    z_i| x_i)^\top,\nabla_{\text{vech1}(\Sigma)} \log p(   y_i,    z_i|  {\theta},X_i)^\top\right)^\top,\\
    \nabla_\theta^2 \log p_\theta(   y_i,    z_i| x_i)=& \begin{bmatrix} \nabla_{\beta}^2 \log p_\theta(   y_i,    z_i| x_i) & \nabla_{\text{vech1}(\Sigma),\beta} \log p_\theta(   y_i,    z_i|x_i)^\top \\  \nabla_{\text{vech1}(\Sigma),\beta} \log p_\theta(   y_i,    z_i| x_i) & \nabla_{\text{vech1}(\Sigma)}^2 \log p_\theta(   y_i,    z_i| x_i) \end{bmatrix},
\end{align*}
with vech1(A) the half-vectorization of the symmetric matrix A after which the first element is eliminated. That is, for a $2\times2$ matrix A, vech1(A)$=(a_{21},a_{22})^\top$. Each of the elements of these expressions can be computed as
\begin{align}\label{eq:Vexpres}
    \nabla_\beta \log p_\theta(y_i, z_i \mid x_i) =& x_i^\top \Sigma^{-1}\eta_i,\\
    \nabla_{\beta}^2 \log p_\theta(y_i, z_i \mid  x_i) =& -x_i^\top \Sigma^{-1} x_i,\notag\\
    \nabla_{\text{vech1}(\Sigma)} \log p_\theta(y_i, z_i \mid  x_i) =& \left(\frac{\partial \Sigma}{\partial\text{vech1}(\Sigma)}\right)^\top \text{vec}\left(-\frac{1}{2}(\Sigma^{-1}-\Sigma^{-1}\eta_i\eta_i^\top\Sigma^{-1})\right),\notag\\  
    \nabla_{\text{vech1}(\Sigma)}^2 \log p_\theta(y_i, z_i \mid  x_i) =& \left(\frac{\partial \Sigma}{\partial\text{vech1}(\Sigma)}\right)^\top \frac{1}{2}(-I_{J^2}+\Sigma^{-1}\eta_i\eta_i^\top \otimes I_J + I_J \otimes \Sigma^{-1}\eta_i\eta_i^\top)\notag\\
    &\times(-\Sigma^{-1} \otimes \Sigma^{-1})\frac{\partial \Sigma}{\partial\text{vech1}(\Sigma)},\notag\\
    \nabla_{\text{vech1}(\Sigma),\beta} \log p_\theta(y_i, z_i \mid  x_i) =& (\nabla_{\text{vech1}(\Sigma),\beta_1} \log p_\theta(y_i, z_i \mid  x_i),\dots,\nabla_{\text{vech1}(\Sigma),\beta_r} \log p_\theta(y_i, z_i \mid x_i)), \notag
\end{align}
where $\nabla_{\Sigma,\beta_k} \log p_\theta(y_i, z_i \mid  x_i)= \text{vec1}(A_k+A_k^\top-\text{diag}(A_k))$ with $A_k=-\Sigma^{-1} x_{ik} \eta_i^\top \Sigma^{-1}$.

Third, we derive expressions for $\nabla_\theta \log p_\theta(   y_i,    z_i| x_i)$ and $\nabla_\theta^2 \log p_\theta(   y_i,    z_i| x_i)$ in \eqref{eq:Vhat} with respect to the parameters $\xi$ used in the SEGA algorithm, instead of $\Sigma$:
\begin{align*}
    \nabla_\theta \log p_\theta(   y_i,    z_i| x_i)=&\left(\nabla_\beta \log p_\theta(   y_i,    z_i|  x_i)^\top,\nabla_{\xi} \log p_\theta(   y_i,    z_i|  x_i)^\top\right)^\top,\\
    \nabla_\theta^2 \log p_\theta(   y_i,    z_i|  x_i)=& \begin{bmatrix} \nabla_{\beta}^2 \log p_\theta(   y_i,    z_i|  x_i) & \nabla_{\xi,\beta} \log p_\theta(   y_i,    z_i| x_i)^\top \\  \nabla_{\xi,\beta} \log p_\theta(   y_i,    z_i|  x_i) & \nabla_{\xi}^2 \log p_\theta(   y_i,    z_i|  x_i) \end{bmatrix},
\end{align*}
where $\nabla_{\xi} \log p_\theta(   y_i,    z_i|  x_i)^\top$ is derived in \eqref{eq:Gexpres}, and $\nabla_\beta \log p_\theta(   y_i,    z_i|  x_i)^\top$ and $\nabla_{\beta}^2 \log p_\theta(   y_i,    z_i|  x_i)$ in \eqref{eq:Vexpres}.
We use the following identity for the derivation of $\nabla_{\xi}^2 \log p_\theta(   y_i,    z_i|  x_i)$:
\begin{equation}
    \frac{\partial^2 A}{\partial B\partial B} = \left[I_{n_B}\otimes\frac{\partial A}{\partial C}\right]\frac{\partial^2C}{\partial B\partial B}+\left[\frac{\partial C}{\partial B}^\top\otimes I_{n_A}\right]\frac{\partial^2 A}{\partial C\partial C}\frac{\partial C}{\partial B},
\end{equation}
where \(A=A(C(B))\), with \(A\in\mathbb R^{n_A}\), \(B\in\mathbb R^{n_B}\), and \(C\in\mathbb R^{n_C}\). If \(A\), \(B\), or \(C\) are matrices, the derivatives are understood with respect to their vectorized representations, so that \(n_A\), \(n_B\), and \(n_C\) denote the corresponding vector dimensions. We can now write
\begin{align*}
    &\nabla_{\xi}^2 \log p_\theta( y_i,  z_i|x_i)  = \frac{\partial\left[\frac{\partial}{\partial \xi}\log p_\theta( y_i,  z_i|x_i)\right]}{\partial\xi}
     = \frac{\partial\left[\frac{\partial}{\partial \Sigma^{-1}}\log p_\theta( y_i,  z_i|x_i)\frac{\partial\Sigma^{-1}}{\partial\xi}\right]}{\partial\xi}\\
    & = \left[I_{n_\xi}\otimes\frac{\partial}{\partial \Sigma^{-1}}\log p_\theta( y_i,  z_i|x_i)\right]\frac{\partial\text{vec}(\frac{\partial\Sigma^{-1}}{\partial\xi})}{\partial\xi}+ \left[\frac{\partial\Sigma^{-1}}{\partial\xi}^\top\otimes I_{1}\right]\frac{\partial\text{vec}(\frac{\partial}{\partial \Sigma^{-1}}\log p_\theta(y_i, z_i|x_i))}{\partial\xi}\\
    & = \left[I_{n_\xi}\otimes\frac{\partial}{\partial \Sigma^{-1}}\log p_\theta( y_i,  z_i|x_i)\right]\frac{\partial^2\Sigma^{-1}}{\partial\xi\partial\xi}+\frac{\partial\Sigma^{-1}}{\partial\xi}^\top\left\{\frac{\partial\left(\frac{\partial}{\partial \Sigma^{-1}}\log p_\theta( y_i, z_i|x_i)\right)}{\partial\xi}\right\}\\
    & = \left[I_{n_\xi}\otimes\frac{\partial}{\partial \Sigma^{-1}}\log p_\theta( y_i,  z_i|x_i)\right]\frac{\partial^2\Sigma^{-1}}{\partial\xi\partial\xi}+\frac{\partial\Sigma^{-1}}{\partial\xi}^\top\left\{\frac{\partial\left(\frac{\partial}{\partial \Sigma^{-1}}\log p_\theta( y_i,  z_i|x_i)\right)}{\partial\Sigma^{-1}}\frac{\partial\Sigma^{-1}}{\partial{\xi}}\right\}\\
    & = \left[I_{n_\xi}\otimes\frac{\partial}{\partial \Sigma^{-1}}\log p_\theta( y_i,  z_i|x_i)\right]\frac{\partial^2\Sigma^{-1}}{\partial\xi\partial\xi}+\frac{\partial\Sigma^{-1}}{\partial\xi}^\top\frac{\partial^2}{\partial \Sigma^{-1}\partial\Sigma^{-1}}\log p_\theta( y_i,  z_i|x_i)\frac{\partial\Sigma^{-1}}{\partial{\xi}}\\
    & = \left[I_{n_\xi}\otimes\left(0.5\text{vec}(\Sigma)^\top-0.5{\eta}_i^\top\otimes{\eta}_i^\top\right)\right]\frac{\partial^2\Sigma^{-1}}{\partial\xi\partial\xi}+\frac{\partial\Sigma^{-1}}{\partial\xi}^\top\left[-0.5\Sigma\otimes\Sigma\right]\frac{\partial\Sigma^{-1}}{\partial{\xi}}.
\end{align*}
Consider the dimensions of the final line: $(n\times nJ^2)(nJ^2 \times n)+(n\times J^2)(J^2 \times J^2)(J^2 \times n)$. Now we require expressions for:
\begin{align}
    \frac{\partial^2\Sigma^{-1}}{\partial\xi\partial\xi} &= \left[I_{n_\xi}\otimes\frac{\partial\Sigma^{-1}}{\partial \kappa}\right]\frac{\partial^2 \kappa}{\partial\xi\partial\xi}+\left[\frac{\partial \kappa}{\partial \xi}^\top\otimes I_{n_\Sigma}\right]\frac{\partial^2\Sigma^{-1}}{\partial \kappa\partial \kappa}\frac{\partial \kappa}{\partial \xi},\\
    \frac{\partial^2\Sigma^{-1}}{\partial\kappa\partial\kappa} &= \left[I_{n_\kappa}\otimes\frac{\partial\Sigma^{-1}}{\partial \text{vech}(C)}\right]\frac{\partial^2 \text{vech}(C)}{\partial\kappa\partial\kappa}+\left[\frac{\partial \text{vech}(C)}{\partial \kappa}^\top\otimes I_{n_\Sigma}\right]\frac{\partial^2\Sigma^{-1}}{\partial \text{vech}(C)\partial \text{vech}(C)}\frac{\partial \text{vech}(C)}{\partial \kappa}. \notag
\end{align}
Here,   $\frac{\partial^2\Sigma^{-1}}{\partial \text{vech}(C)\partial \text{vech}(C)} = \left[D^\top\otimes I_{J^2}\right]\frac{\partial^2\Sigma^{-1}}{\partial C\partial C}D$, where $D$ is the duplication matrix such that $\text{vec}(C) = D\text{vech}(C)$.
The matrix $\frac{\partial^2\Sigma^{-1}}{\partial C\partial C}$ has dimensions $J^4 \times J^2$ and equals $I_{J^2}\otimes(I_{J^2}+K_{J^2,J^2}) $. To construct the derivative $\frac{\partial^2 \kappa}{\partial\xi\partial\xi}$, first note that $\frac{\partial \kappa}{\partial\xi} = \text{diag}\left(b_1\phi({\xi_1}),\dots,b_{n_\xi}\phi({\xi_{n_\xi}})\right)$. We can show that $\text{vec}(\frac{\partial \kappa}{\partial\xi}) = D(\left(b_1\phi({\xi_1}),\dots,b_{n_\xi}\phi({\xi_{n_\xi}})\right)^\top)$, where $D$ is an appropriate selection matrix. We can then show $\frac{\partial^2 \kappa}{\partial\xi\partial\xi} = D\times\text{diag}(-{\xi_1}b_1\phi({\xi_1}),\dots,-\xi_{n_\xi}b_{n_\xi}\phi({\xi_{n_\xi}}))$.
 The matrix $\frac{\partial^2C}{\partial \kappa\partial \kappa}$ has dimensions $n_{\kappa}J^2 \times n_{\kappa}$. The elements of $\frac{\partial^2 \text{vech}(C)}{\partial\kappa\partial\kappa}$ are constructed as $\frac{\partial^2 \text{vech}(C)_l}{\partial\kappa_j\partial\kappa_k} = \frac{\partial^2 \psi_{l}(\bm{\kappa})}{\partial \kappa_{j} \partial \kappa_{k}}$ with
\begin{align*}
\frac{\partial^2 \psi_{l}(\bm{\kappa})}{\partial \kappa_{j} \partial \kappa_{k}}=
\left\{
\begin{array}{ll}
-\sqrt{J}\cos\left(\kappa_{l}\right)\prod_{s\in\{1,\dots,l-1\}}\sin\left(\kappa_{s}\right) & \text{if } j=k<l \text{ and } l<d_C,\\
\sqrt{J}\cos\left(\kappa_{j}\right)\cos\left(\kappa_{k}\right)\cos\left(\kappa_{l}\right)\prod_{s\in\{1,\dots,l-1\}\backslash \{j,k\}}\sin\left(\kappa_{s}\right) & \text{if } j\neq k<l \text{ and } l<d_C,\\
-\sqrt{J}\cos\left(\kappa_{j}\right)\prod_{s\in\{1,\dots,l\}\backslash j}\sin\left(\kappa_{s}\right) & \text{if } j<k=l \text{ and } l<d_C,\\
-\sqrt{J}\cos\left(\kappa_{k}\right)\prod_{s\in\{1,\dots,l\}\backslash \{k\}}\sin\left(\kappa_{s}\right) & \text{if } k\leq j=l \text{ and } l<d_C,\\
-\sqrt{J}\prod_{s\in\{1,\dots,l-1\}}\sin\left(\kappa_{s}\right) & \text{if } j=k<l \text{ and } l=d_C,\\
\sqrt{J}\cos\left(\kappa_{j}\right)\cos\left(\kappa_{k}\right)\prod_{s\in\{1,\dots,l-1\}\backslash \{j,k\}}\sin\left(\kappa_{s}\right) & \text{if } j\neq k<l \text{ and } l=d_C,\\
0 & \text{if otherwise.} 
\end{array}
\right.
\end{align*}
The cross Hessian is given as $\nabla_{\beta,\xi}\log \phi_{J}\left({z}_i;x_i{\beta},\Sigma\right) = \left(x_i^\top\otimes{\eta}_i^\top\right) \frac{\partial \Sigma^{-1}}{\partial C}\frac{\partial C}{\partial{\kappa}}\frac{\partial {\kappa}}{\partial{\xi}}$.

\subsection{Additional details and results empirical application}\label{A:mnp_application}
\subsubsection{Computing elasticities} 
For given parameter values $\theta$ and covariates $x_i$, where the $k$-th row of $x_i$ is specified as $x_{ik} = (e_k^\top,\, \log p_{ik} - \log p_{i0})$ with $p_{ij}$ the price of alternative $j$, the cross-price elasticity of alternative $j$ with respect to the price of alternative $k$ is defined as
\begin{align}
e_{jk}(\theta; x_i)
= \frac{\partial {\Pr}_\theta(Y = j \mid x_i)}{\partial \log p_{ik}}
= \frac{\partial {\Pr}_\theta(Y = j \mid x_i)}{\partial x_{ik,J+1}}.
\end{align}
Since $\Pr_\theta(Y = j \mid x_i)$ is not available in closed form, neither is the corresponding elasticity. However, we can approximate the choice probability as $\widehat{\Pr}_\theta(Y = j \mid x_i)= \frac{1}{S} \sum_{s=1}^S p(y_i=j \mid z_{i,s})$,where $z_{i,s} \sim N(x_i \beta, \Sigma)$ . 

An estimate of the elasticity can then be obtained via a symmetric finite-difference approximation. Evaluating the probability at $x_i$ with the elements corresponding to the price difference for alternative $k$ equal to $[x_{ik}]_{J+1} + \epsilon$ and $[x_{ik}]_{J+1} - \epsilon$ yields
\begin{align}
    \widehat{e}_{jk}(\theta; x_i)= \frac{\widehat{\Pr}_\theta(Y = j \mid x_{ik,J+1} + \epsilon)-\widehat{\Pr}_\theta(Y = j \mid x_{ik,J+1} - \epsilon)}{2\epsilon}.
\end{align}
Note that as $S\rightarrow\infty$ this estimate gets arbitrarily accurate. We report the elasticity estimates evaluated at $\theta = \bar{\theta}_n$. 

We construct confidence intervals for the elasticities by drawing $\theta^{(m)} \sim N\!\left(\bar{\theta}_n, \frac{1}{n}\widehat{V}\right)$, and transforming each draw into the corresponding elasticity, $\widehat{e}_{jk}\!\left(\theta^{(m)}, x_i\right)$. These transformed draws are then used to summarize the sampling variability of the elasticity estimates.
To ensure positive definiteness and the identification constraints for the covariance matrix, the draws $\theta^{(m)}$ are from the unconstrained parameter space with $\beta$ and $\xi$.

\subsubsection{Additional empirical results}
Figure~\ref{fig:conv} shows the trajectories of the coefficients in $\beta$ and the transformed angles in $\xi$ in the SEGA algorithm. The algorithm takes approximately 59 hours.

\begin{figure}[!ht]
\caption{MNP parameter trajectories in SEGA with 100,000 iterations.}
\centering
\includegraphics*[width=\textwidth]{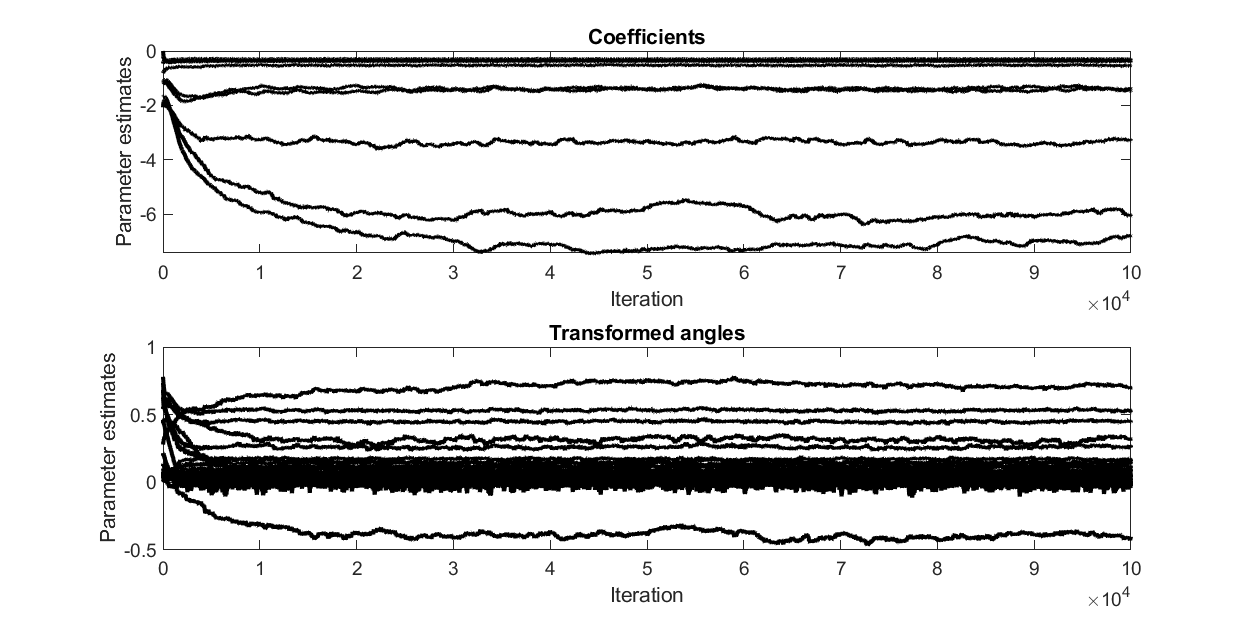}
\label{fig:conv}
\end{figure}

Table~\ref{tab:mnp_coef_estimates} shows the coefficient and variance estimates, together with their standard errors. The first two columns show that price has a statistically significant negative impact on purchase probabilities. Moreover, the estimates for the intercepts show statistically significant differences across purchase frequencies of the different pasta brands. The final two columns show that the purchase price explains different portions of the utility of each choice alternative.

% tables with estimation results
\begin{table}[!ht]
  \centering 
  \caption{Multinomial probit coefficient and variance estimates}
  \begin{threeparttable}
    \begin{tabular}{lrrrr}
    \toprule \toprule
          & \multicolumn{2}{c}{Coefficients} & \multicolumn{2}{c}{Variances} \\
          \cmidrule(lr){2-3}  \cmidrule(lr){4-5}
          & \multicolumn{1}{l}{Estimate} & \multicolumn{1}{l}{Std. Error} & \multicolumn{1}{l}{Estimate} & \multicolumn{1}{l}{Std. Error} \\
          \midrule
    Price & -0.335 & 0.002 &       &  \\
    Mueller & -0.392 & 0.004 &       &  \\
    Creamette & -0.410 & 0.007 & 1.080 & 0.014 \\
    Ronzoni & -0.396 & 0.009 & 0.972 & 0.0140 \\
    San Giorgio & -0.584 & 0.024 & 1.161 & 0.031 \\
    No Yolks & -1.517 & 0.031 & 2.309 & 0.047 \\
    Hodgson Mills & -1.593 & 0.045 & 2.092 & 0.064 \\
    Healthy Harvest & -7.845 & 0.067 & 15.600 & 0.221 \\
    DaVinci & -3.777 & 0.221 & 3.326 & 0.293 \\
    Dececco & -6.687 & 0.130 & 8.882 & 0.291 \\
    \bottomrule \bottomrule
    \end{tabular}%
\begin{tablenotes}
\footnotesize
\item This table shows the estimates and robust standard errors for the elements in $\beta$ and on the diagonal of $\Sigma$. The pasta brand names indicate the intercepts (columns 2-3) and variances (columns 4-5) in the different utility equations. Note that the utilities are differenced with the utility of Barilla, and that the variance of the utility for Mueller is fixed at one. 
\end{tablenotes}
\end{threeparttable}
  \label{tab:mnp_coef_estimates}
\end{table}

Table~\ref{tab:mnp_coef_covariances} shows the covariance estimates in the multinomial probit model, together with their standard errors. We find statistically significant covariances across the unobserved components of the utilities of the mostly purchased pasta brands. However, due to the fact that utilities are differenced with the base category, the interpretation of the individual covariance estimates is challenging. 

\begin{table}[!h]
  \centering \tiny
  \caption{Multinomial probit covariance estimates}
  \begin{threeparttable}
    \begin{tabular}{clrrrrrrrr}
    \toprule \toprule
          &       & \multicolumn{1}{l}{Mueller} & \multicolumn{1}{l}{Creamette} & \multicolumn{1}{l}{Ronzoni} & \multicolumn{1}{l}{San Giorgio} & \multicolumn{1}{l}{No Yolks} & \multicolumn{1}{l}{Hodgson Mills} & \multicolumn{1}{l}{Healthy Harvest} & \multicolumn{1}{l}{DaVinci} \\
          \midrule
   \multirow{8}[0]{*}{\begin{sideways}Estimates\end{sideways}} & Creamette & 0.378 &       &       &       &       &       &       &  \\
          & Ronzoni & 0.451 & 0.399 &       &       &       &       &       &  \\
          & San Giorgio & 0.445 & 0.518 & 0.445 &       &       &       &       &  \\
          & No Yolks & 0.377 & 0.452 & 0.355 & 0.410 &       &       &       &  \\
          & Hodgson Mills & 0.400 & 0.394 & 0.418 & 0.428 & 0.183 &       &       &  \\
          & Healthy Harvest & 0.010 & 0.095 & 0.116 & 0.046 & 0.219 & 0.023 &       &  \\
          & DaVinci & 0.096 & 0.182 & 0.253 & 0.295 & 0.191 & 0.223 & -0.019 &  \\
          & Dececco & 0.026 & 0.031 & 0.128 & 0.088 & 0.069 & 0.073 & 0.111 & 0.178 \\
          \midrule
    \multirow{8}[0]{*}{\begin{sideways}Standard Errors\end{sideways}} & Creamette & 0.007 &       &       &       &       &       &       &  \\
          & Ronzoni & 0.005 & 0.007 &       &       &       &       &       &  \\
          & San Giorgio & 0.007 & 0.008 & 0.006 &       &       &       &       &  \\
          & No Yolks & 0.011 & 0.012 & 0.009 & 0.027 &       &       &       &  \\
          & Hodgson Mills & 0.010 & 0.011 & 0.008 & 0.028 & 0.027 &       &       &  \\
          & Healthy Harvest & 0.067 & 0.066 & 0.050 & 0.081 & 0.139 & 0.113 &       &  \\
          & DaVinci & 0.060 & 0.035 & 0.026 & 0.105 & 0.082 & 0.109 & 0.406 &  \\
          & Dececco & 0.057 & 0.056 & 0.047 & 0.119 & 0.130 & 0.149 & 0.548 & 0.565 \\
          \bottomrule \bottomrule
    \end{tabular}%
\begin{tablenotes}
\footnotesize
\item This table shows the estimates and robust standard errors for the elements on the off-diagonal of $\Sigma$. The pasta brand names indicate the covariances across the different utility equations. Note that the utilities are differenced with the utility of Barilla.
\end{tablenotes}
\end{threeparttable}
  \label{tab:mnp_coef_covariances}
\end{table}

\subsection{Exact maximum likehood estimation}\label{A:mnp_MLE}
In the three choice MNP model, we can evaluate the likelihood function exactly via a one dimensional numerical integration. The log-likelihood is
\begin{align}
\ell_n(\theta)=\sum_{i=1}^n\sum_{j=0}^2{I}\{y_i=j\}\log P_{ij}(\theta),
\end{align}
where \(P_{ij}(\theta)=\Pr(y_i=j\mid x_i)\). For $z_i\sim N(x_i\beta,\Sigma),$ the probability of $j=0$ is $P_{i0}(\theta)=\Phi_2(-x_i\beta;\Sigma)$. For alternative \(j=1\), $P_{i1}(\theta)=\Pr(z_{i1}\ge \max\{0,z_{i2}\})$,
which is evaluated using the conditional Gaussian representation
\begin{align}
P_{i1}(\theta)=\int_{-\infty}^{\infty}\phi(z_2;x_{i2}^\top\beta,\Sigma_{22})\left[1-\Phi\!\left(\frac{\max(z_2,0)-\mu_1(z_2)}{\sqrt{\tilde\Sigma_1}}\right)\right]dz_2,
\end{align}
where $\mu_1(z_2)=x_{i1}^\top\beta+\frac{\Sigma_{12}}{\Sigma_{22}}(z_2-x_{i2}^\top\beta)$, and $\tilde\Sigma_1=\Sigma_{11}-\frac{\Sigma_{12}^2}{\Sigma_{22}}$. Similarly, for alternative \(j=2\), $P_{i2}(\theta)=\Pr(z_{i2}\ge \max\{0,z_{i1}\})$, which is evaluated as
\begin{align}
P_{i2}(\theta)=\int_{-\infty}^{\infty}\phi(z_1;x_{i1}^\top\beta,\Sigma_{11})\left[1-\Phi\!\left(\frac{\max(z_1,0)-\mu_2(z_1)}{\sqrt{\tilde\Sigma_2}}\right)\right]dz_1,
\end{align}
where $\mu_2(z_1)=x_{i2}^\top\beta+\frac{\Sigma_{12}}{\Sigma_{11}}(z_1-x_{i1}^\top\beta)$, and $\tilde\Sigma_2=\Sigma_{22}-\frac{\Sigma_{12}^2}{\Sigma_{11}}$. The univariate integrals required to evaluate $P_{i1}(\theta)$ and $P_{i2}(\theta)$ can be computed arbitrarily accurate using off-the-shelf numerical integration methods.
 
\section{Example 2: Random effects Tobit model}

\subsection{SEGA implementation}\label{A:tobit_fisher}
First, we derive the exact expressions in the gradient of the log augmented likelihood with respect to  presented in \eqref{eq:tobit_sega}:
\begin{align*}
\nabla_\theta \log p_\theta(y_{1:n},z_{1:n}|x_{1:n})=\sum_{i=1}^n\left(\sum_{t=1}^T\nabla_{\beta,c}\log \phi_1(y_{it}^*;\mu_{it},\sigma^2)^\top,\nabla_{\delta}\log\phi_r(\alpha_i;0_r,\Omega)^\top\right)^\top,
\end{align*}
where $\mu_{it}=h_{it}^\top\beta+w_{it}^\top\alpha_i$ and $\nabla_{\beta,c}\log \phi_1(y_{it}^*;h_{it}^\top\beta+w_{it}^\top\alpha_i,\sigma^2)=$
\begin{align}
 \left(\nabla_\beta\log \phi_1(y_{it}^*;h_{it}^\top\beta+w_{it}^\top\alpha_i,\sigma^2)^\top,\nabla_c\log \phi_1(y_{it}^*;h_{it}^\top\beta+w_{it}^\top\alpha_i,\sigma^2)\right)^\top.
\end{align}
The three required gradient components can be evaluated analytically as
\begin{align}\label{eq:grad_tobit_sega}
&\nabla_\beta\log \phi_1(y_{it}^*;h_{it}^\top\beta+w_{it}^\top\alpha_i,\sigma^2) = \frac{1}{\sigma^2}(y_{it}^*-h_{it}^\top\beta-w_{it}^\top\alpha_i)h_{it},\\
&\nabla_{c}\log \phi_1(y_{it}^*;h_{it}^\top\beta+w_{it}^\top\alpha_i,\sigma^2) = -\frac{1}{2}+\frac{1}{2\sigma^2}(y_{it}^*-h_{it}^\top\beta-w_{it}^\top\alpha_i)^2,\notag\\
&\nabla_{\delta}\log\phi_r(\alpha_i;0_r,\Omega) = -\frac{1}{2}\text{vec}(\Omega^{-1})^\top(I_{r^2}+K_{r,r})(\tilde{D}\otimes I_r)\frac{\partial \tilde{D}}{\partial \delta}+\left(\alpha_i^\top\Omega^{-1}\tilde{D}\otimes\alpha_i^\top\Omega^{-1}\right)\frac{\partial \tilde{D}}{\partial \delta},\notag
\end{align}
with $\frac{\partial \tilde{D}}{\partial \delta} = \frac{\partial \tilde{D}}{\partial D}\frac{\partial D}{\partial \delta}$. Here, $\frac{\partial \tilde{D}}{\partial D}$ is a diagonal matrix with elements $\frac{\partial \tilde{D}}{\partial D}_{(i-1)r+1,(i-1)r+1} = \exp(\delta_{i,i})$ for $i = 1,\dots,r$, and ones in the remaining diagonal elements. Also, we have that  $\frac{\partial D}{\partial \delta} = P$, where $P$ is a sparse matrix such that $\text{vec}(D) = P\delta$.

Second, we produce draws from $p_\theta(z_{1:n}|d_{1:n})$ in \eqref{eq:tobit_sega} using 
a Gibbs sampling scheme that iteratively employs two steps:
\begin{itemize}
    \item[] Step 1: Generate one draw from $p_{\theta}(y_{1:n}^*|d_{1:n},\alpha_{1:n}) = \prod_{i=1}^{n}\prod_{t=1}^{T}p({y}_{it}^*|{d}_{it},\alpha_i)$,
    \item[] Step 2: Generate one draw from $p_{\theta}(\alpha_{1:n}|d_{1:n},y_{1:n}^*)=\prod_{i=1}^n p_\theta({\alpha}_i|d_i,y_{1:n}^*)$,
\end{itemize}
where $d_{it} = (y_{it},h_{it}^\top,w_{it}^\top)^\top$, and $d_{i} = (d_{i1}^\top,\dots,d_{iT}^\top)$.
For observations such that $y_{it}>0$, generating from 
${y}_{it}^*$ in Step 1 is equivalent to setting ${y}_{it}^* = y_{it}$. When $y_{it} = 0$, then one must generate from the truncated normal $p_\theta({y}_{it}^*|{d}_{it},{\alpha}_i) = \phi({y}_{it}^*;{h}_{it}^\top{\beta}+{w}_{it}^\top{\alpha}_i,\sigma^2)/\Phi(0;{h}_{it}^\top{\beta}+{w}_{it}^\top{\alpha}_i,\sigma^2)I({y}_{it}^*\le 0)$.
The random effects in Step 2 can be generated independently across observations from  the multivariate normal distribution $p_\theta({\alpha}_i|d_i,y_{1:n}^*) = \phi_{r}(\alpha_i;\bar{{\alpha}}_i,V_i)$ with
\begin{align*}
    \bar{{\alpha}}_i =\frac{1}{\sigma^2}V_i\left(\sum_{t=1}^{T}w_{it}{y}_{it}^*-\left[\sum_{t=1}^{T}w_{it}h_{it}^\top\right]\beta\right) \text{ and } V_i = \left[\Omega^{-1}+\frac{1}{\sigma^2}\sum_{t=1}^{T} w_{it}w_{it}^\top\right]^{-1}.
\end{align*}

\subsection{Variance estimator implementation}\label{A:Tobitse}
First, we use a two step generation process based on the representation $p_\theta(z_i|d_i) = p_\theta(\alpha_i|d_i,y_i^*)p_\theta(y_i^*|d_i)$ to generate i.i.d. draws from $p_\theta(z_{1:n}|d_{1:n})$.  

In the first step, we generate from the multivariate truncated normal  
\begin{align}
p_\theta(y_i^{*}|d_i) = \frac{\phi_{T}(y_i^{*};H_i\beta,\Sigma_i)\prod_{t=1}^T\mathbb{I}\left[y_{it} = y_{it}^*\mathbb{I}(y_{it}^*>0)\right]}{\int \phi_{T}(y_i^{*};H_i\beta,\Sigma_i)\prod_{t=1}^T\mathbb{I}\left[y_{it} = y_{it}^*\mathbb{I}(y_{it}^*>0)\right]},
\end{align}
where  $H_i = (h_{i1},\dots,h_{it})^\top$, $W_i = (w_{i1},\dots,w_{it})^\top$ and $\Sigma_i = W_i\Omega W_i+\sigma^2 I_T$. To generate from this distribution we can  condition on the values of $y_{it}^*=y_{it}>0$ that are observed and do not need to be generated. Denote $y_{i,u}^* = \{y_{it}^*: y_{it}=0\}$ to be   the set of unobserved variables and $y_{i,o}^* = \{y_{it}^*: y_{it}>0\}$ those that are observed. Using the properties of normal distributions we can show that
\begin{align}
p_\theta(y_{i,u}^{*}|y_{i,o}^*,d_i) = \frac{\phi_{T_u}(y_{i,u}^{*};\mu_{i,u},\Sigma_{i,uu}-\Sigma_{i,ou}\Sigma_{i,oo}^{-1}\Sigma_{i,uo})\prod_{y_{it}^*\in y_{i,u}^{*}} \mathbb{I}(y_{it}^*<0)}{\int \phi_{T_u}(y_{i,u}^{*};\mu_{i,u},\Sigma_{i,uu}-\Sigma_{i,ou}\Sigma_{i,oo}^{-1}\Sigma_{i,uo})\prod_{y_{it}^*\in y_{i,u}^{*}} \mathbb{I}(y_{it}^*<0)}
\end{align}
where $T_u$ denotes the number of elements in $y_{i,u}^*$,  $\mu_{i,u} = H_{i,u}\beta+\Sigma_{i,ou}\Sigma_{i,oo}^{-1}(y_{i,o}-H_{i,o}\beta)$, $\Sigma_{i,ou} = \text{cov}_\theta(y_{i,o}^{*},y_{i,u}^{*}|x_i)$, $\Sigma_{i,uu} =\text{var}_\theta(y_{i,u}^{*}|x_i)$  and $\Sigma_{i,oo} =\text{var}_\theta(y_{i,o}^{*}|x_i)$. This conditional distribution is a $T_u$-variate truncated normal, from which we can draw independent draws using the method by \cite{botev2017normal}.

In the second step, we generate from $p_\theta({\alpha}_i|d_i,y_{1:n}^*)$, which is the Gaussian distribution also used in the Gibbs sampling scheme in the SEGA algorithm.

Second, we derive the expressions for $\nabla_\theta \log p_\theta(   y_i,    z_i| x_i)$ and $\nabla_\theta^2 \log p_\theta(   y_i,    z_i| x_i)$ in \eqref{eq:Vhat} to compute standard errors for the estimates of $\beta$, $\sigma^2$, and the lower triangular elements of $\Omega$, which can be extracted as $\text{vech}(\Omega)$:
\begin{align}
\nabla_{\theta}\ell(\theta,z_i) =& \left(\nabla_{\beta}\ell(\theta,z_i)^\top,\nabla_{\sigma^2}\ell(\theta,z_i)^\top,\nabla_{\text{vech}(\Omega)}\ell(\theta,z_i)^\top\right)^\top,\\
\nabla_\theta^2 \ell(\theta,z_i) =& \begin{bmatrix} 
        \nabla_\beta^2 \ell(\theta,z_i) & \nabla_{\beta,\sigma^2} \ell(\theta,z_i) & {0}_{p\times r(r+1)/2}  \\
        \nabla_{\beta,\sigma^2} \ell(\theta,z_i)^\top & \nabla_{\sigma^2}^2 \ell(\theta,z_i)  & {0}_{1\times r(r+1)/2}  \\
        {0}_{ r(r+1)/2\times p} & {0}_{ r(r+1)/2\times 1} &\nabla_{\text{vech}(\Omega)}^2 \ell(\theta,z_i).
\end{bmatrix}.
\end{align}
The following are the terms required to evaluate the gradient and Hessian:
\begin{align}\label{eq:grad_tobit}
    &i)\hspace{0.3cm}\nabla_{\beta}\ell(\theta,z_i) = \frac{1}{\sigma^2}\sum_{t=1}^T\eta_{it}^*h_{it}; \quad ii)\hspace{0.3cm}  \nabla_{\sigma^2} \ell(\theta,z_i) = -\frac{T}{2\sigma^2}+\frac{1}{2(\sigma^2)^2}\sum_{t=1}^T{\eta_{it}^*}^2;\\
    &iii)\hspace{0.3cm}\nabla_{\text{vech}(\Omega)} \ell(\theta,z_i) = \left[\frac{\partial \Omega}{\partial\text{vech}(\Omega)}\right]^\top\text{vec}\left[-\frac{1}{2}(\Omega^{-1}-\Omega^{-1}\alpha_i\alpha_i^\top\Omega^{-1})\right]; \notag\\
    &iv)\hspace{0.3cm}\nabla_{\beta}^2 \ell(\theta,z_i) = -\frac{1}{\sigma^2}\sum_{t=1}^T h_{it} h_{it}^\top; \hspace{0.5cm} v)\hspace{0.3cm} \nabla_{\beta,\sigma^2} \ell(\theta,z_i)= -\frac{1}{(\sigma^2)^2}\sum_{t=1}^T\eta_{it}^*h_{it}; \notag\\
    &vi)\hspace{0.3cm}\nabla_{\sigma^2}^2 \ell(\theta,z_i) = \frac{T}{2(\sigma^2)^2}-\frac{1}{(\sigma^2)^3}\sum_{t=1}^T{\eta_{it}^*}^2; \notag\\
    & vii)\hspace{0.3cm}\nabla_{\text{vech}(\Omega)}^2 \ell(\theta,z_i) = P^\top\frac{1}{2}(-I_{r^2}+\Omega^{-1} \alpha_i  \alpha_i^\top \otimes I_r +I_r \otimes \Omega^{-1}\alpha_i  \alpha_i^\top)(-\Omega^{-1} \otimes \Omega^{-1})P, \notag
\end{align}
where $\eta_{it}^* = y_{it}^*-h_{it}^\top\beta-w_{it}^\top\alpha_i$ and $P$ is the matrix such that $\text{vec}(\Omega) = P\text{vech}(\Omega)$.

Third, we derive expressions for $\nabla_\theta \log p_\theta(   y_i,    z_i| x_i)$ and $\nabla_\theta^2 \log p_\theta(   y_i,    z_i| x_i)$ in \eqref{eq:Vhat} with respect to the parameters $c$ and $\delta$, instead of $\sigma^2$ and $\Omega$:
\begin{align}
\nabla_{\theta}\ell(\theta,z_i) =& \left(\nabla_{\beta}\ell(\theta,z_i)^\top,\nabla_{c}\ell(\theta,z_i)^\top,\nabla_{\delta}\ell(\theta,z_i)^\top\right)^\top,\\
\nabla_\theta^2 \ell(\theta,z_i) =& \begin{bmatrix} 
        \nabla_\beta^2 \ell(\theta,z_i) & \nabla_{\beta,c} \ell(\theta,z_i) & {0}_{p\times r(r+1)/2}  \\
        \nabla_{\beta,c} \ell(\theta,z_i)^\top & \nabla_{c}^2 \ell(\theta,z_i)  & {0}_{1\times r(r+1)/2}  \\
        {0}_{ r(r+1)/2\times p} & {0}_{ r(r+1)/2\times 1} &\nabla_{\delta}^2 \ell(\theta,z_i).
\end{bmatrix}.
\end{align}
where $\nabla_{\beta}\ell(\theta,z_i)$ and $\nabla^2_{\beta}\ell(\theta,z_i)$ are derived in \eqref{eq:grad_tobit} and $\nabla_{c}\ell(\theta,z_i)$ and $\nabla_{\delta}\ell(\theta,z_i)$ in \eqref{eq:grad_tobit_sega}.
%\dn{HAVE TO CHECK THIS AND IMPLEMENT THIS:}
We have 
\begin{align}
    \nabla_{\beta c}\ell(\theta,z_i)&=-\frac{1}{\exp(c)}\sum_{t=1}^T\eta_{it}^*h_{it}, \quad \nabla_{c}^2\ell(\theta,z_i)=-\frac{1}{2\exp(c)}\sum_{t=1}^T(\eta_{it}^*)^2,\\
\nabla_{\delta\delta}^2\ell(\theta,z_i)&=J_\Omega(\delta)^\top H_\Omega(\theta,z_i)J_\Omega(\delta)+R_\Omega(\theta,z_i),
\end{align}
where
\begin{align}
J_\Omega(\delta)=&(I_{r^2}+K_{r,r})(\widetilde D\otimes I_r)\frac{\partial \operatorname{vec}(\widetilde D)}{\partial \delta},\\
    H_\Omega(\theta,z_i)=&\frac{1}{2}\left[-I_{r^2}+\Omega^{-1}\alpha_i\alpha_i^\top\otimes I_r+I_r\otimes \Omega^{-1}\alpha_i\alpha_i^\top\right](-\Omega^{-1}\otimes \Omega^{-1}),\\
\left[R_\Omega(\theta,z_i)\right]_{ab}=&(-\frac{1}{2}\left\{\Omega^{-1}-\Omega^{-1}\alpha_i\alpha_i^\top\Omega^{-1}\right\})^\top\operatorname{vec}\left(\frac{\partial^2\Omega}{\partial\delta_a\partial\delta_b}\right).
\end{align}
where \(K_{r,r}\) is the commutation matrix. The matrix \(\partial\operatorname{vec}(\widetilde D)/\partial\delta\) is sparse, with entries equal to \(\exp(\delta_{jj})\) for diagonal Cholesky elements and one for off-diagonal Cholesky elements.

\subsection{Additional details and results empirical application}\label{A:tobit_application}

\subsubsection{Computing probabilities and marginal effects}
For given parameter values \(\theta\), covariates \(x_{it}\), and random effects
\(\alpha_i\), the probability of positive sales is
\begin{align}
{\Pr}_\theta(y_{it}>0\mid x_{it},\alpha_i)={\Pr}_\theta(y_{it}^*>0\mid x_{it},\alpha_i)=\Phi\left(\frac{\mu_{it}}{\sigma}\right),
\end{align}
with $\mu_{it}=h_{it}^{\top}\beta+w_{it}^{\top}\alpha_i$. The conditional expectation of observed sales is
\begin{align}
\E_\theta(y_{it}\mid x_{it},\alpha_i)=\mu_{it}\Phi\left(\frac{\mu_{it}}{\sigma}\right)+\sigma\phi\left(\frac{\mu_{it}}{\sigma}\right).
\end{align}
The marginal effect of the price index on expected observed sales is
\begin{align}
{ME}_{p}(\theta)=\frac{\partial \E_\theta(y_{it}\mid x_{it},\alpha_i)}{\partial p_{it}}=\Phi\left(\frac{\mu_{it}}{\sigma}\right)\left(\beta_p+\alpha_{pi}\right),
\end{align}
where we write $\alpha_i=(\alpha_{0i},\alpha_{pi},\alpha_{\ell i})^\top$, with the subscripts denote the random intercept, price-index coefficient, and lag-sales coefficient. Similarly, let \(\beta_p\) and \(\beta_\ell\) denote the fixed coefficients on the price index and lagged sales. 

To construct Figure~\ref{fig:tobit}, we evaluate the model on a grid of household-specific lag-sales coefficients $\lambda_i=\beta_\ell+\alpha_{\ell i}$. Let $\lambda_g\in\left[Q_{0.01}\{N(\beta_\ell,\Omega_{\ell\ell})\},Q_{0.99}\{N(\beta_\ell,\Omega_{\ell\ell})\}\right]$, with $g=1,\ldots,G$ and \(Q_\tau\{F\}\) denotes the \(\tau\)-quantile of distribution \(F\). For each grid point, set $a_{\ell,g}=\lambda_g-\beta_\ell$. Since the random effects are jointly Gaussian, the conditional mean of the remaining random effects given \(\alpha_{\ell i}=a_{\ell,g}\) is
\begin{align}
\E_\theta\left[\begin{pmatrix}
\alpha_{0i}\\
\alpha_{pi}
\end{pmatrix}
\bigg|\alpha_{\ell i}=a_{\ell,g}\right]=\frac{1}{\Omega_{\ell\ell}}
\begin{pmatrix}
\Omega_{0\ell}\\
\Omega_{p\ell}
\end{pmatrix}
a_{\ell,g}.
\end{align}
We construct \(\widehat P_g(\theta)\) and \(\widehat{ME}_{p,g}(\theta)\) by evaluating ${\Pr}_\theta(y_{it}>0\mid x_{it},\alpha_i)$ and ${ME}_{p}(\theta)$ at \(\theta=\bar\theta_n\) and $\alpha_i=\alpha_g$ with 
\begin{align}
\alpha_g=\left(\frac{\Omega_{0\ell}}{\Omega_{\ell\ell}}a_{\ell,g},\frac{\Omega_{p\ell}}{\Omega_{\ell\ell}}a_{\ell,g},a_{\ell,g}\right)^\top .
\end{align}

Confidence intervals are constructed by drawing $\theta^{(m)}\sim N\left(\bar\theta_n,\frac{1}{n}\widehat V\right)$, and transforming each draw into the corresponding constrained parameters \((\beta^{(m)},\sigma^{2(m)},\Omega^{(m)})\), and recomputing \(\widehat P_g(\theta^{(m)})\) and \(\widehat{ME}_{p,g}(\theta^{(m)})\) for each grid point. The pointwise confidence intervals are obtained from the empirical quantiles of these transformed draws.

\subsubsection{Additional empirical results}
Figure~\ref{fig:conv_tobit} shows the trajectories for the estimates in the top panel of Table~\ref{tab:tobit_estimates}. The total computation time across all iterations is approximately 75 hours.

\begin{figure}[!ht]
\caption{Tobit parameter trajectories in SEGA with 100,000 iterations.}
\centering
\includegraphics*[width=\textwidth]{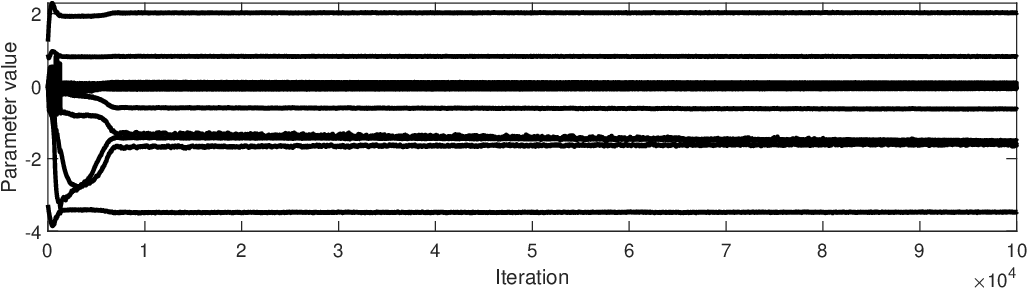}
\label{fig:conv_tobit}
\end{figure}

Table~\ref{tab:tobit_estimates} shows the estimates for the parameters in the Tobit model together with their standard errors. The price index has a negative impact on the sales, while lagged sales has a positive coefficient. We find evidence for strong unobserved individual heterogeneity, with the variance of the random effects to be statistically significantly different from zero. According to the estimation results, the random effects are also correlated. Finally, we find that all promotion coefficients are statistically significantly different from zero.

\begin{table}[!ht]
  \centering \tiny
  \caption{Tobit parameter estimates}
  \begin{threeparttable}
    \begin{tabular}{lrrrrrrrr}
        \toprule \toprule
    \multicolumn{3}{c}{Coefficients} & \multicolumn{3}{c}{Variance Random Effects} & \multicolumn{3}{c}{Covariance Random Effects} \\
    \cmidrule(lr){1-3}\cmidrule(lr){4-6}\cmidrule(lr){7-9}
    Variable & \multicolumn{1}{l}{Estimate} & \multicolumn{1}{l}{Std. Error} & \multicolumn{1}{l}{Random Effect} & \multicolumn{1}{l}{Estimate} & \multicolumn{1}{l}{Std. Error} & \multicolumn{1}{l}{Random Effects} & \multicolumn{1}{l}{Estimate} & \multicolumn{1}{l}{Std. Error} \\
    \midrule
    Intercept & -3.476 & 0.005 & \multicolumn{1}{l}{Intercept} & 0.040 & 0.002 & \multicolumn{1}{l}{Intercept/Price index} & 0.002 & 0.001 \\
    Price index & -0.019 & 0.017 & \multicolumn{1}{l}{Price index} & 0.051 & 0.003 & \multicolumn{1}{l}{Intercept/Lag sales} & -0.126 & 0.002 \\
    Lag Sales & 2.025 & 0.007 & \multicolumn{1}{l}{Lag Sales} & 0.439 & 0.005 & \multicolumn{1}{l}{Price index/Lag sales} & -0.012 & 0.002 \\
          &       &       & \multicolumn{1}{l}{Error term} & 2.266 & 0.016 &       &       &  \\
          &       &       &       &       &       &       &       &  \\
    \multicolumn{3}{c}{Feature coefficients} & \multicolumn{3}{c}{Display coefficients I} & \multicolumn{3}{c}{Display coefficients II} \\
    \cmidrule(lr){1-3}\cmidrule(lr){4-6}\cmidrule(lr){7-9}
    Feature & \multicolumn{1}{l}{Estimate} & \multicolumn{1}{l}{Std. Error} & \multicolumn{1}{l}{Display} & \multicolumn{1}{l}{Estimate} & \multicolumn{1}{l}{Std. Error} & \multicolumn{1}{l}{Display} & \multicolumn{1}{l}{Estimate} & \multicolumn{1}{l}{Std. Error} \\
    \midrule
    Back Page & 1.326 & 0.035 & \multicolumn{1}{l}{Front End Cap} & 2.020 & 0.025 & \multicolumn{1}{l}{Rear End Cap} & 2.012 & 0.016 \\
    Front Page & 1.850 & 0.025 & \multicolumn{1}{l}{In-Aisle} & 2.171 & 0.025 & \multicolumn{1}{l}{Secondary Location Display} & 2.226 & 0.027 \\
    Interior Page & 1.998 & 0.018 & \multicolumn{1}{l}{In-Shelf} & 2.063 & 0.013 & \multicolumn{1}{l}{Side-Aisle End Cap} & 2.018 & 0.133 \\
    Wrap Back & 1.169 & 0.036 & \multicolumn{1}{l}{Mid-Aisle End Cap} & 2.091 & 0.033 & \multicolumn{1}{l}{Store Front} & 1.790 & 0.054 \\
    Wrap Front & 0.999 & 0.040 & \multicolumn{1}{l}{Promo/Seasonal Aisle} & 2.375 & 0.025 & \multicolumn{1}{l}{Store Rear} & 2.075 & 0.035 \\
    Wrap Interior & 1.157 & 0.062 &       &       &       &       &       &  \\
    \bottomrule \bottomrule
    \end{tabular}%
\begin{tablenotes}
\footnotesize
\item This table shows the estimates and robust standard errors for the elements in $\beta$ and $\Omega$ and for $\sigma^2$. Note that we omitted the coefficients for the week dummies in $\beta$.
\end{tablenotes}
\end{threeparttable}
  \label{tab:tobit_estimates}
\end{table}

\subsection{Numerical experiments}\label{A:tobit_simulation}
We consider the random effects Tobit model described in Section~\ref{sec:examplesTobit}. The data is generated according to \eqref{Eq:linktobit} and \eqref{Eq:yystar}. We specify $h_{it} = w_{it} = (1, x_{it})^\top$, where $x_{it} \sim N_r(0, I_r)$. We examine two designs. First, we consider a low-dimensional random effects specification with $r=2$, $\beta_0 = (0,3,2)^\top$, $\sigma_0^2 = 1$, and
\[
\Omega_0 =
\begin{pmatrix}
1 & 0.2 & -0.1 \\
0.2 & 1 & 0.4 \\
-0.1 & 0.4 & 1
\end{pmatrix}.
\]
Second, we consider a higher-dimensional random effects structure with $r=8$, $\beta_0 = (-0.65,\,1.18,\,-0.76,\,-1.11,\,-0.85,\,-0.57,\,-0.56,\,0.18,\,-0.2)^\top$, $\sigma_0^2 = 1$, and
\[
\Omega_0 =
\begin{pmatrix}
1 & 0.3 & -0.2 & -0.2 & 0.4 & -0.1 & 0 & 0.3 & -0.1 \\
0.3 & 1 & 0 & -0.3 & 0.4 & -0.1 & 0 & 0.4 & 0 \\
-0.2 & 0 & 1 & 0 & -0.3 & 0.1 & -0.1 & 0.1 & 0.1 \\
-0.2 & -0.3 & 0 & 1 & -0.3 & 0.2 & 0.1 & 0.1 & -0.3 \\
0.4 & 0.4 & -0.3 & -0.3 & 1 & 0 & 0 & 0.5 & 0.2 \\
-0.1 & -0.1 & 0.1 & 0.2 & 0 & 1 & 0.1 & -0.1 & 0.2 \\
0 & 0 & -0.1 & 0.1 & 0 & 0.1 & 1 & -0.1 & -0.3 \\
0.3 & 0.4 & 0.1 & 0.1 & 0.5 & -0.1 & -0.1 & 1 & -0.1 \\
-0.1 & 0 & 0.1 & -0.3 & 0.2 & 0.2 & -0.3 & -0.1 & 1
\end{pmatrix}.
\]
For both model specifications, we fix the group size at $T=100$ and consider two sample sizes: a small-sample setting with $n=100$ groups and a large-sample setting with $n=1,000$ groups. For each design, we conduct a repeated sampling experiment by generating $1{,}000$ independent datasets from the corresponding data-generating process and applying SEGA to each sample. 

Table~\ref{tab:MCstatTobit} reports the Monte Carlo summary statistics. The average absolute bias across all model parameters is negligible in all experiments and decreases further as the sample size increases. The RMSE exhibits a clear decline as $n$ grows from $100$ to $1{,}000$, consistent with the expected improvement in estimator precision with larger samples. This pattern holds in both the $r=2$ and $r=8$ settings, confirming the strong finite-sample performance of SEGA. Coverage rates are also close to the nominal level. Across parameters, the minimum, median, mean, and maximum coverage values all lie near 0.95, indicating that the uncertainty quantification provided by SEGA is accurate.

% \begin{figure}[!h]
% \caption{Monte Carlo distributions coefficient estimates}
% \centering
% \includegraphics*[width=\textwidth]{figures_simulation/MCdens28Tobit.eps}
% \begin{flushleft}
% \footnotesize
% This figure shows the Monte Carlo distributions of the estimates for $\beta_{r+1}$ by SEGA in the solid gold line. The panels correspond to Monte Carlo experiments with $p=2,8$ and $n=100,1000$. The vertical lines indicate the parameter value in the data generating process.
% \end{flushleft}
% \label{fig:MCdens28Tobit}
% \end{figure}

\begin{table}[H]
  \centering
  \caption{Monte Carlo summary statistics}
  \begin{threeparttable}
    \begin{tabular}{lrrrrrr}
    \toprule \toprule
    & \multicolumn{6}{c}{\textbf{Panel A: $r=2$}} \\
    \cmidrule(lr){2-7}
    $n$ &  &  & \multicolumn{4}{c}{Coverage} \\
    \cmidrule(lr){4-7}
       &  Absolute bias            &  RMSE    & Min & Median & Mean & Max \\
    \midrule
    100  &  0.0038 & 0.1106 &      0.919  &  0.942 &  0.941 &   0.955
\\
    1000 &  0.0007 & 0.0350 &    0.931  &  0.940 &   0.943  &  0.953
\\
    \midrule
    & \multicolumn{6}{c}{\textbf{Panel B: $r=8$}} \\
    \cmidrule(lr){2-7}
    $n$ &  &  & \multicolumn{4}{c}{Coverage} \\
    \cmidrule(lr){4-7}
       &   Absolute bias           &  RMSE    & Min & Median & Mean & Max \\
    \midrule
    100  & 0.0037  &  0.1114 &  0.920 &   0.946 &   0.945  &  0.964  \\
    1000 & 0.0008 &   0.0352 &     0.918  &  0.950  &  0.947  &  0.970
 \\
    \bottomrule \bottomrule
    \end{tabular}
    \begin{tablenotes}
    \footnotesize
    \item This table reports absolute bias, root mean squared error, and coverage statistics (minimum, median, mean, and maximum coverage across parameters) across Monte Carlo replications.
    \end{tablenotes}
  \end{threeparttable}
  \label{tab:MCstatTobit}
\end{table} 
\section{Supplementary lemma’s}\label{A:lemmas}
\subsection{Supplementary lemma’s proofs}
The following lemmas are used to prove results in the main paper. The first result is given in \cite{chen2020statistical}.

\begin{lemma}{[Lemma B.2, \cite{chen2020statistical}]}\label{lem:chen2020}
    Let $z_k$ be a sequence in $\mathbb{R}_+$ that satisfies the recursion
    $$ z_k \leq (1-\lambda\eta_k)z_{k-1} + D \eta_k^{2+\gamma}, $$
    where $D,\lambda>0$, and $\gamma\ge0$ are fixed constants, and the sequence $\eta_{k}$ is decreasing. Then for any $m \leq k-1$,
    \begin{flalign*}
    z_k %\leq%& \exp(-\lambda\sum_{i=m+1}^{k}\eta_i)+D[\eta_{m}]^{1+\gamma}\lambda^{-1}
    \le \exp\{-\lambda(k-m)\eta_{k}\}z_m+D[\eta_{m}]^{1+\gamma}\lambda^{-1}
    \end{flalign*}
\end{lemma}

\begin{lemma}\label{lem:concavity}
If Assumption \ref{ass:conv} is satisfied, then, for $\mu>0$, 
\begin{flalign*}
  \langle {m}_n(\theta),\theta-\hat\theta \rangle    \le-\frac{\mu}{2}\|\theta-\hat\theta\|^2_2 .
\end{flalign*}
\end{lemma}
\begin{proof}
Since \(\ell_n\) is \(\mu\)-strongly concave, for any \(\theta_1,\theta_2\in\Theta\) we have  
\[
\ell_n(\theta_2)
\;\le\;
\ell_n(\theta_1)
\;+\;\bigl\langle \nabla_\theta \ell_n(\theta_1),\,\theta_2-\theta_1\bigr\rangle
\;-\;\frac {n\mu}{2}\,\|\theta_2-\theta_1\|_2^2.
\]
Set
$\theta_1 =\theta$ and $\theta_2 =\hat\theta$, 
where \(\hat\theta\) satisfies \(m_n(\hat\theta)=0\).  Plugging into the strong‐concavity inequality gives
\[
\ell_n(\hat\theta)
\;\le\;
\ell_n(\theta)
\;+\;\bigl\langle \nabla_\theta \ell_n(\theta),\,\hat\theta-\theta\bigr\rangle
\;-\;\frac{n\mu}{2}\,\|\hat\theta-\theta\|_2^2.
\]
Rearrange:
\[
\bigl\langle \nabla_\theta \ell_n(\theta),\,\hat\theta-\theta\bigr\rangle
\;\ge\;
\ell_n(\hat\theta)-\ell_n(\theta)
\;+\;\frac{n\mu}{2}\,\|\hat\theta-\theta\|_2^2.
\]
But \(\hat\theta\) maximizes \(\ell_n\), so \(\ell_n(\hat\theta)\ge\ell_n(\theta)\).  Hence
\[
\bigl\langle \nabla_\theta \ell_n(\theta),\,\hat\theta-\theta\bigr\rangle
\;\ge\;
0
\;+\;\frac{n\mu}{2}\,\|\hat\theta-\theta\|_2^2
\;=\;
\frac{n\mu}{2}\,\|\theta-\hat\theta\|_2^2.
\]
Equivalently,
\[
\bigl\langle m_n(\theta),\,\theta-\hat\theta\bigr\rangle
\;\le\;
-\frac{\mu}{2}\,\|\theta-\hat\theta\|_2^2.
\]
\end{proof}
 
\subsection{Assumptions MNP example}\label{A:ass_mnp}
In this section, we verify Assumptions \ref{ass:mds}-\ref{ass:conv} in the MNP example. For simplicity, we consider the parametrization with respect to the precision matrix instead of the transformed angles $\xi$. We make the following low-level regularity conditions. 
\begin{assumption}\label{ass:MNP_assumptions} The parameter is $\theta=(\beta^\top,\mathrm{vech}(\Sigma^{-1})^\top)^\top$, and the parameter space is $\Theta=\mathcal{B}\times\mathcal{S}$. The following conditions are satisfied. 
\begin{enumerate}
    \item The parameter space $\Theta$ is compact. For each $i\in\{1,\dots,n\}$, $x_i\in\mathcal{X}$, with $\mathcal{X}$ bounded. 
    \item Each $\Sigma\in\mathcal{S}$ is positive-definite; i.e., $0<\underline{\lambda}\le \lambda_{\min}(\Sigma)\le\lambda_{\max}(\Sigma)\le \overline\lambda<\infty$.
    \item The choice probabilities are bounded below: for some $\underline{p}>0$, $\inf_{i\le n}\inf_{\theta\in\Theta}\Pr_\theta(y_i\mid x_i)\ge \underline{p}$. 
    \item For each $i\le n$ and $k\ge 1$, $z_i^{(k)}\mid d_i,\mathcal F_k\sim p_{\theta}(z_i\mid d_i)$. 
    \item For some $\delta>0$, and all $\theta\in\mathcal{N}_\delta(\hat\theta_n)$, there is an $n'$ large enough such that, for all $n\ge n'$, $-\nabla^2_\theta\ell_n(\tilde\theta)\succeq n\mu\cdot I$, with $\mu>0$.
\end{enumerate}
\end{assumption}
\begin{lemma}\label{lemma:MNP}
Under the regularity conditions in Supplemental Appendix~\ref{A:lemmas}, Assumptions~\ref{ass:mds}-\ref{ass:conv} are satisfied for the MNP model.   
\end{lemma}

\noindent\textbf{Verification of Assumption \ref{ass:mds}}

Define $\eta_i=z_i^{(k)}-x_i\beta$, where we suppress dependence on $k$. First derive the bounds on $\E[\|\eta_i\|_2^2\mid \mathcal{F}_{k}]$ and $\E[\|\eta_i\|_2^4\mid \mathcal{F}_{k}]$ required for verifying the assumptions.  We use the triangle inequality to write $\|\eta_i\|_2^2 \leq (\|z_i^{(k)}\|_2+\|x_i\beta\|_2)^2$. Note that, by Assumption \ref{ass:MNP_assumptions}(1),  $\|x_i\beta\|_2$ is bounded for all $i\le n$. 

Consider $w_i \sim N(x_i\beta,\Sigma)$ with $\E[\|w_i\|_2^2]=\text{trace}(\Sigma)+\|x_i\beta\|^2_2$. It follows that
\begin{align}
    \E[\|w_i\|_2^2] =& \E[\|w_i\|_2^2|w_i \in A]\mathbb{P}[w_i \in A] + \E[\|w_i\|_2^2|w_i \notin A]\mathbb{P}[w_i \notin A] \\
    \geq& \E[\|w_i\|_2^2|w_i \in A]\mathbb{P}[w_i \in A],
\end{align}
Then it follows that $\E[\|z_i\|_2^2\mid \mathcal{F}_{k}]$ is bounded if $\text{trace}(\Sigma)$ and  $\|x_i\beta\|^2_2$ are bounded and the choice probability $\mathbb{P}[w_i \in A]$ is bounded away from zero (i.e., Assumption \ref{ass:MNP_assumptions}(1)-(3)). 

Suppress the dependence on the filtration and write 
\begin{align}
    \E[\|z_i\|_2^4] =& \E\left[ \left(\sum_{j}z_{ij}^2\right)^2 \right] = \sum_j \E[z_{ij}^4]+2\sum_{j<k}\E[z_{ij}^2z_{ik}^2]\\
    \leq& C\left( \sum_j \E[z_{ij}^2]^2+2\sum_{j<k}\E[z_{ij}^2]\E[z_{ik}^2]\right)=C
    \left( \sum_j \E[z_{ij}^2]^2\right)^2=C\left( \E[\|z_{ij}\|_2^2]\right)^2, \notag
\end{align}
where we use that $\E[z_{ij}^4]<C(\E[z_{ij}^2])^2$ with $C\geq 1$, and the Cauchy–Schwarz inequality $\E[z_{ij}^2z_{ik}^2]\leq \sqrt{\E[z_{ij}^4]\E[z_{ik}^4]}$. Since $\E[\|z_i\|_2^2]$ is bounded, $\E[\|z_i\|_2^4]$ is bounded. Finally, note that 
% \begin{align}
%     \widehat{m}_n(\theta) = \frac{1}{n}\sum_{i=1}^n ((x_i^\top \Sigma^{-1} \eta_i)^\top,\frac{1}{2}\text{vec}(\Sigma-\eta_i\eta_i^\top))^\top.
% \end{align}
\begin{align}
    g_i(\theta,z_i) = ((x_i^\top \Sigma^{-1} \eta_i)^\top,\frac{1}{2}\text{vec}(\Sigma-\eta_i\eta_i^\top))^\top.
\end{align}

\noindent\textbf{Verify Assumption 1(i).} This follows from the Fisher's identity under Assumption \ref{ass:MNP_assumptions}(4).

\noindent\textbf{Verify Assumption 1(ii)}
\begin{align}
    \E\left[\|g_i(\theta^{(k)}_n;z^{(k)}_{i})-\E[g_i(\theta^{(k)}_n;z^{(k)}_{i})\mid\mathcal{F}_{k}]\|_2^2\mid\mathcal{F}_{k}\right]\le \E\left[\|g_i(\theta^{(k)}_n;z^{(k)}_{i})\|^2_2\mid\mathcal{F}_{k}\right]\\
    \le \E\left[\| x_i^\top \Sigma^{-1} \eta_i\|_2^2+\|\frac{1}{2}\text{vec}(\Sigma-\eta_i\eta_i^\top)\|_2^2\mid\mathcal{F}_{k}\right]\\
    \le \E\left[\| \| x_i\|_2^2\| \Sigma^{-1}\|_2^2\| \eta_i\|_2^2+\frac{1}{2}(\|\Sigma\|_F^2 +\|\eta_i\|_2^4)\mid\mathcal{F}_{k}\right]\le \nu^2,
\end{align}
where we use the bounds on $\E[\|\eta_i\|_2^2\mid \mathcal{F}_{k}]$ and $\E[\|\eta_i\|_2^4\mid \mathcal{F}_{k}]$, and on $\| x_i\|_2^2$ independently of $i$, $\| \Sigma^{-1}\|_2^2$, and $\|\Sigma\|_F^2$ (Assumption \ref{ass:MNP_assumptions} (1)\&(3)).

\noindent\textbf{Verify Assumption 2(i)} 
For the MNP we have
\begin{align}\label{eq:l-2i}
    \ell_n(\theta) = \sum_{i=1}^{N}\log\text{Pr}(Y_i=y_i|{\theta},x_i)=\sum_{i=1}^{N}\log p_i(\theta),
\end{align}
with
\begin{align*}
    \text{Pr}(Y_i=j|{\theta},x_i) =& 
  \int I(Z_{ij}\geq \max(\max(Z_{i}),0))\phi_{J}\left({z};x_i{\beta},\Sigma\right) d{Z}_{i} \text{ for } j>0,\\
   \text{Pr}(Y_i=0|{\theta},x_i) =& 
  \int I(\max(Z_{i})< 0) \phi_{J}\left({z};x_i{\beta},\Sigma\right) d{Z}_{i}.
\end{align*}
For $\theta_2,\theta_1\in\Theta^2$, the exact second-order Taylor expansion of $\ell_n\left(\theta\right)$ around $\theta_1$ is 
\begin{align*}
    \ell_n\left(\theta_2\right) = \ell_n\left(\theta_1\right) +\left\langle\nabla_\theta \ell_n\left(\theta_1\right), \theta_2-\theta_1\right\rangle+\frac{1}{2}\left(\theta_1-\theta_2\right)^\top\nabla_\theta^2 \ell_n(\tilde{\theta})\left(\theta_1-\theta_2\right),
\end{align*}
with $\tilde{\theta}$ on the line segment between $\theta_1$ and $\theta_2$. From the expression in \eqref{eq:l-2i} follows
\begin{align}
    \nabla_\theta \ell_n({\theta}) =& \sum_i \frac{1}{p_i(\theta)} \nabla_\theta p_i(\theta),\\
    \nabla_\theta^2 \ell_n({\theta}) =& \sum_i \frac{1}{p_i(\theta)} \nabla_\theta^2 p_i(\theta) - \frac{1}{p_i(\theta)^2} \nabla_\theta p_i(\theta)\nabla_\theta p_i(\theta)^\top. \label{eq:ll-2i}
\end{align}
Now, by Assumption \ref{ass:MNP_assumptions}(5), for $n$ large enough, we have that 
\begin{align*}
    \left(\theta_1-\theta_2\right)^\top\nabla_\theta^2 \ell_n(\tilde{\theta})\left(\theta_1-\theta_2\right)\leq -n\mu \|\theta_1-\theta_2\|_2^2.
\end{align*}
Using the Taylor expansion and this result, we now have that 
\begin{align*}
\ell_n\left(\theta_2\right) &= \ell_n\left(\theta_1\right) +\left\langle\nabla_\theta \ell_n\left(\theta_1\right), \theta_2-\theta_1\right\rangle+\frac{1}{2}\left(\theta_1-\theta_2\right)^\top\nabla_\theta^2 \ell_n(\tilde{\theta})\left(\theta_1-\theta_2\right)\\
&\le\ell_n\left(\theta_1\right) +\left\langle\nabla_\theta \ell_n\left(\theta_1\right), \theta_2-\theta_1\right\rangle-\frac{n}{2}\mu \|\theta_1-\theta_2\|_2^2.
\end{align*}

\noindent\textbf{Verify Assumption 2(ii)} 
By the mean value theorem, for any $\theta \in \Theta$, there exists a line-by-line intermediate value $\tilde{\theta}$ such that:
$ m_n(\theta)-m_n(\widehat\theta_n)=\frac{1}{n}\nabla_\theta^2\ell_n(\tilde{\theta})(\theta-\widehat\theta_n)$.
Use the sub-multiplicative property to write
\begin{align}\label{eq:mm-22}
    \|m_n(\theta)-m_n(\widehat\theta_n)\|_2
\le\left\|\frac{1}{n}\nabla_\theta^2\ell_n(\widetilde\theta)\right\|_2\|\theta-\widehat\theta_n\|_2.
\end{align}
Use the expression in \eqref{eq:ll-2i} and the triangle inequality and the sub-multiplicative property to write
\begin{align}
  \left\|  \frac{1}{n}\nabla_\theta^2\ell_n(\theta)\right\|_2 \le \frac{1}{n}\sum_{i=1}^n\left\{\frac{\|\nabla_\theta^2p_i(\theta)\|_2}{p_i(\theta)}+\frac{\|\nabla_\theta p_i(\theta)\|_2^2}{p_i(\theta)^2}
\right\}.
\end{align}
Under Assumptions 6(1)–(2), the MNP choice probabilities are twice continuously differentiable in $\theta$. Since the parameter space is compact and the covariates lie in a bounded set, there exist finite constants $C_1$ and $C_2$ such that $\|\nabla_\theta p_i(\theta)\|_2
\le C_1$ and
$\|\nabla_\theta^2p_i(\theta)\|_2
\le C_2$ for $i=1,\dots,n$, and $\theta \in \Theta$. Under Assumption 6(3), it further holds that $p_i(\theta) \geq \underline p$ so that
\begin{align}
\left\|\frac{1}{n}\nabla_\theta^2\ell_n(\theta)\right\|_2\le\frac{C_2}{\underline p}+\frac{C_1^2}{\underline p^2}=L.
\end{align}
Substituting this bound into \eqref{eq:mm-22} yields $\|m_n(\theta)-m_n(\widehat\theta_n)\|_2\le L\|\theta-\widehat\theta_n\|_2$.

\end{document}